\documentclass{article}
\usepackage{arxiv}

\usepackage{amssymb, amsmath, amsthm}
\usepackage{pifont}
\usepackage[polish,english]{babel}
\usepackage[T1,nomathsymbols]{polski}

\usepackage{url}
\usepackage[hidelinks]{hyperref}
\usepackage[small]{caption}
\usepackage{graphicx}
\usepackage{amsmath}
\usepackage{booktabs}
\usepackage{algorithm}
\usepackage{algorithmic}
\usepackage{tkz-graph}
\usepackage{tkz-berge}
\usepackage{tikz}
\usetikzlibrary{positioning,arrows,fit,calc}
\usetikzlibrary{decorations.pathmorphing}
\usetikzlibrary{shapes}

\usepackage{tikz}
\usetikzlibrary{fit}
\usetikzlibrary{calc}
\usetikzlibrary{shapes,decorations,shadows,calc}
\usetikzlibrary{decorations.pathmorphing, decorations.markings,decorations.pathreplacing,decorations.shapes}

\usetikzlibrary{patterns}
\usepgflibrary{arrows}
\usetikzlibrary{trees}
\usetikzlibrary{backgrounds, arrows, spy}
\usetikzlibrary{circuits.logic.US}

\tikzset{>=latex,
	point/.style = {circle,draw,thick,minimum size=2mm,inner sep=0pt},
}

\newcommand{\slin}[3]
{
\scriptsize
\draw[-, very thick] ($({#1}) +(0,0.09)$) -- node[above=0.1] {#2}  node[below=0.1] {#3} ($({#1}) +(0,-0.09)$);
}

\usepackage{pifont}

\usepackage{multirow}

\newtheorem{theorem}{Theorem}

\newtheorem{lemma}[theorem]{Lemma}
\newtheorem{corollary}[theorem]{Corollary}
\newtheorem{example}[theorem]{Example}

\newfont{\midmathxx}{cmsy10 scaled 1440}
\newfont{\bigmathxx}{cmsy10 scaled 1440}
\newfont{\smallmathxx}{cmsy10 scaled 720}

\newcommand{\OWLQL}{\textsl{OWL\,2\,QL}}

\newcommand{\NL}{\textsc{NL}}

\newcommand{\ACz}{{\ensuremath{\textsc{AC}^0}}}
\newcommand{\NCo}{{\ensuremath{\textsc{NC}^1}}}

\newcommand{\coNP}{\textsc{coNP}}
\newcommand{\NP}{\textsc{NP}}
\newcommand{\PTime}{\textsc{P}}

\DeclareFontFamily{U}{MnSymbolC}{}
\DeclareSymbolFont{MnSyC}{U}{MnSymbolC}{m}{n}
\DeclareMathSymbol{\diamondminus}{\!}{MnSyC}{120}
\DeclareMathSymbol{\diamondplus}{\!}{MnSyC}{"7C}
\DeclareFontShape{U}{MnSymbolC}{m}{n}{
    <-5>  MnSymbolC4
    <5-6>  MnSymbolC5
   <6-7>  MnSymbolC6
   <7-8>  MnSymbolC7
   <8-9>  MnSymbolC8
   <9-10> MnSymbolC9
  <10-12> MnSymbolC10
  <12->   MnSymbolC12}{}

\newcommand{\LTL}{\textsl{LTL}}
\newcommand{\MTL}{\textsl{MTL}}
\newcommand{\hMTL}{\textsl{hornMTL}}
\newcommand{\cMTL}{\textsl{coreMTL}}

\newcommand{\dMTL}{\textsl{datalogMTL}}

\newcommand{\diamondmin}{\diamondminus\hspace*{-0.5mm}}

\newcommand{\range}{\varrho}

\newcommand{\D}{\mathcal{D}}
\newcommand{\A}{\mathcal{A}}

\newcommand{\I}{\mathcal{I}}

\newcommand{\C}{\mathfrak C}

\newcommand{\cl}{\mathsf{cl}}
\newcommand{\type}{{\boldsymbol{t}}}
\newcommand{\types}{{\boldsymbol{s}}}
\newcommand{\trace}{{\boldsymbol{tr}}}
\newcommand{\q}{{\boldsymbol{q}}}
\newcommand{\Q}{{\boldsymbol{Q}}}

\newcommand{\tem}{{\mathsf{ts}}}

\newcommand{\slit}{\sigma}
\newcommand{\tlit}{\tau}
\newcommand{\lit}{\lambda}

\newcommand{\inr}{\mathsf{in}}
\newcommand{\suc}{\mathsf{suc}}

\title{Data Complexity and Rewritability of Ontology-Mediated Queries in Metric Temporal Logic under the Event-Based Semantics}

\author{Vladislav Ryzhikov\\
Birkbeck, University of London, UK\\
\texttt{vlad@dcs.bbk.ac.uk}\\
\And
Przemyslaw Andrzej Walega\\
University of Oxford, UK and University of Warsaw, Poland\\
\texttt{przemyslaw.walega@cs.ox.ac.uk}\\
\And
Michael Zakharyaschev\\
Birkbeck, University of London, UK\\
\texttt{michael@dcs.bbk.ac.uk}
}

\hypersetup{draft}

\begin{document}

\maketitle

\begin{abstract}
We investigate the data complexity of answering queries mediated by ontologies given in metric temporal logic \MTL{} under the event-based semantics assuming that data instances are finite timed words with binary fractions as timestamps. We identify classes of ontology-mediated queries answering which can be done in \ACz{}, \NCo, L, \NL, \PTime{}, and \coNP{} for data complexity, provide their rewritings to first-order logic and its extensions with primitive recursion, transitive closure or datalog, and establish lower complexity bounds.
\end{abstract}

\section{Introduction}

In this paper, we are concerned with the following problem: given a formula $\Pi$ of metric temporal logic \MTL{} and an atomic proposition $A$, is it possible to construct a query $\Q(x)$ in some standard query language such that, for any data instance $\D$ of atoms timestamped with  binary fractions and any timestamp $t$ from $\D$, we have $\Pi,\D \models A(t)$ iff $\Q(t)$ is true in $\D$?

\MTL{} was originally designed for modelling and reasoning about real-time systems \cite{DBLP:journals/rts/Koymans90,DBLP:journals/iandc/AlurH93}; for a survey see~\cite{DBLP:reference/mc/BouyerFLMO018}. Recently, combinations of \MTL{} with description logics have been suggested as temporal ontology languages~\cite{DBLP:conf/ecai/Gutierrez-Basulto16,DBLP:conf/frocos/BaaderBKOT17}. Datalog with \MTL-operators was used by~\cite{DBLP:journals/jair/BrandtKRXZ18,DBLP:conf/semweb/MehdiKSXKBHRR17} for practical ontology-based access to temporal log data aiming to facilitate detection and monitoring complex events in asynchronous systems based on sensor measurements.
For example, a Siemens turbine has a coast down if the rotor speed was below 1500 in the previous 30 seconds, while no more than 2 minutes before that the speed was above 6600 for 30 seconds. The event `coast down' can be encoded by the following $\MTL$-formula, where $\diamondmin_{(r,s]} \varphi$ ($\boxminus_{(r,s]} \varphi$) is true at a timestamp $t$ if $\varphi$ holds at some (respectively, all) $t'$ with $r < t-t' \le s$:
$$
\boxminus_{(0,30s]} \mathsf{speed}_{< 1500} \land {}
\diamondminus_{(0, 2m]} \boxminus_{(0,30s]} \mathsf{speed}_{> 6600} \to
 \mathsf{cdown}.
$$
%
%
To find when a coast down occurred, a Siemens engineer can now simply execute the query $\mathsf{cdown}(x)$ mediated by an \MTL-ontology with formulas such as the one above, whose atoms are related to sensor data by appropriate mappings.
Answering \dMTL{} queries in the streaming setting was considered by~\cite{WalegaAAAI}.

The underpinning idea of classical ontology-based data access (OBDA) \cite{CDLLR07,IJCAI-18} is a reduction of ontology-mediated query (OMQ) answering to standard database query evaluation. As known from descriptive complexity~\cite{Immerman99}, the existence of such reductions, or \emph{rewritings}, is closely related to the data complexity of OMQ answering, which is by now well understood for atemporal OMQs both uniformly (for all OMQs in a given language) and non-uniformly (for individual OMQs)~\cite{DBLP:journals/tods/GottlobOP14,DBLP:conf/rweb/BienvenuO15,DBLP:journals/tods/BienvenuCLW14,DBLP:conf/ijcai/LutzS17}.


Temporal ontology and query languages have attracted attention of datalog and description logic communities since the 1990s; see~\cite{DBLP:books/bc/tanselCGSS93/BaudinetCW93,DBLP:conf/dagstuhl/ChomickiT98,DBLP:conf/time/LutzWZ08,DBLP:conf/time/ArtaleKKRWZ17} for surveys. In recent years, the proliferation of temporal data from various sources and its importance for analysing the behaviour of complex systems and decision making in all economic sectors have intensified research into formalisms that can be used for querying temporal databases and streaming data~\cite{DBLP:journals/jaise/SoyluGSJKONB17,DBLP:journals/ai/BeckDE18,DBLP:conf/aaai/RoncaKGMH18}. OBDA with atemporal ontologies and query languages with linear temporal logic \LTL{} operators has been in use since~\cite{BBL13,DBLP:conf/rweb/OzcepM14}. Rewritability and data complexity of OMQs in the description logics \textsl{DL-Lite} and $\mathcal{EL}$ extended with \LTL{} operators were considered in~\cite{DBLP:conf/ijcai/ArtaleKKRWZ15,DBLP:conf/ijcai/Gutierrez-Basulto16}.


Here, we investigate the (uniform) rewritability and data complexity problems for basic OMQs given in metric temporal logic \MTL, assuming that data instances are finite sets of atoms timestamped by dyadic rationals and that \MTL{} is interpreted under the event-based semantics where atoms refer to events (state changes) rather than to states themselves~\cite{DBLP:conf/formats/OuaknineW08}. \MTL{} is more succinct, expressive, and versatile compared to \LTL, being able to model both synchronous (discrete) and asynchronous (real-time) settings.

First, we observe that answering arbitrary \MTL-OMQs is \coNP-complete for data complexity (in contrast to \NCo-completeness for \LTL-OMQs).  OMQs in the Horn fragment \hMTL{} are \PTime-complete and rewritable to datalog(FO), which extends datalog with FO-formulas built from EDB predicates; in fact, we establish P-hardness already for the fragment $\cMTL\!^\boxminus$ of \hMTL{} with binary rules (like in \OWLQL) and box operators only. OMQs in $\cMTL\!^\diamondminus$ turn out to be FO(TC)-rewritable (FO with transitive closure) and \NL-hard.
We then classify \MTL-OMQs by the type of ranges $\range$ constraining their temporal operators $\diamondmin_{\range}$ and $\boxminus_{\range}$: infinite   $(r,\infty)$ and $[r,\infty)$, punctual $[r,r]$, and arbitrary non-punctual $\range$. We show that OMQs of the first type are FO-rewritable and can be answered in $\ACz$. OMQs of the second type are FO(RPR)-rewritable (FO with relational primitive recursion) and $\NCo$-complete. For the third type, we obtain an \NL{} upper bound with rewritability to FO(TC)  and $\NCo$ lower bound; for \hMTL-OMQs of this type, the results are improved to L with rewritability to FO(DTC) (FO with deterministic closure).


\section{\MTL{} Ontology-Mediated Queries}\label{sec:MTL}

In the context of event monitoring, we consider a `past' variant of \MTL, which is a propositional modal logic with constrained operators $\diamondmin_\range$ `sometime in the past within range $\range$' and $\boxminus_\range$ `always in the past within range $\range$,\!' interpreted over finite timed words under the event-based semantics.
We assume that timestamps in timed words are given as non-negative dyadic rational numbers (finite binary fractions), the set of which is denoted by $\mathbb Q_2^{\ge 0}$. The ranges $\range$ in $\diamondmin_\range$ and $\boxminus_\range$ are non-empty intervals with end-points
in $\mathbb{Q}_2^{\ge 0} \cup \{\infty\}$.  
%

An \MTL-\emph{program}, $\Pi$, is a finite set of \emph{rules} of the form
\begin{equation}\label{eq:rule}
\vartheta_1 \land \dots \land \vartheta_k \to \vartheta_{k+1} \lor \dots \lor \vartheta_{k+l},
\end{equation}
where each $\vartheta_i$ takes the form $A$, $\diamondmin_{\range} A$, or $\boxminus_{\range} A$, for an atomic proposition $A$. We denote the empty $\land$ by $\top$ (truth) and empty $\lor$ by $\bot$ (falsehood). Using fresh atoms, every \MTL-formula can be transformed to an equivalent (in the sense of giving the same answers to queries) \MTL-program.

An \MTL-program is called a $\hMTL$-\emph{program} if, in all of its rules \eqref{eq:rule}, $l \le 1$ and $\vartheta_{k+1}$ is an atom. As usual, $\vartheta_{k+1}$ is called the \emph{head} of the rule and $\vartheta_1 \land \dots \land \vartheta_k$ its \emph{body}.
A $\hMTL$-program is a $\cMTL$-\emph{program} if $k+l \le 2$.
An $\MTL$- (\hMTL- or \cMTL-) \emph{ontology-mediated query} (OMQ) takes the form $\q = (\Pi, A)$, where $\Pi$ is an $\MTL$- (resp., \hMTL- or \cMTL-) program and $A$ an atom.

Intuitively, a \emph{data instance}, $\D$, can be thought of as a word $\boldsymbol{A}_0(\bar 0),\dots,\boldsymbol{A}_k(\bar k)$ with timestamps $\bar 0 < \dots < \bar k$, $\bar i \in \mathbb{Q}^{\ge 0}_2$, where each $\boldsymbol{A}_i$ is the set of atoms that are true at $\bar i$.
%
%
Formally, we represent $\D$ as the FO-structure
\begin{equation}\label{struc}
\D ~=~ (\varDelta, < , \varTheta , \mathsf{bit}_{\it in}, \mathsf{bit}_{\it fr}, A^\D_1, \dots, A^\D_p),
\end{equation}
with domain $\varDelta = \{0,\dots,\ell\}$ ordered by $<$, \emph{timestamps} $\varTheta = \{0, \dots, k\}$, $1 \le k \le \ell$, and subsets $A^\D_i \subseteq \varTheta$. The ternary predicates $\mathsf{bit}_{\it in}$ and $\mathsf{bit}_{\it fr}$ are such that, for any $n \in \varTheta$ and $i \in \varDelta$, there are unique $b_i,  c_i \in \{0,1\}$ with $\mathsf{bit}_{\it{in}}(n,i,b_i)$ and $\mathsf{bit}_{\it{fr}}(n,i,c_i)$. These predicates give the \emph{value} $\bar n \in \mathbb Q_2^{\ge 0}$ of every timestamp $n \in \varTheta$: $\bar n = b_\ell \dots b_0.c_\ell \dots c_0$ iff $\mathsf{bit}_{\it{in}}(n,i,b_i)$ and $\mathsf{bit}_{\it{fr}}(n,i,c_i)$ hold for all $i \le\ell$. We assume that $\bar n < \bar m$ if $n <m$.
For any $r \in \mathbb{Q}^{\ge 0}_2$, we can define an FO-formula
$\mathsf{dist}_{<r}(x,y)$ that holds in $\mathcal{D}$ iff $x,y \in \varTheta$ and $0 \le \bar x - \bar y < r$, its variants $\mathsf{dist}_{> r}(x,y)$, $\mathsf{dist}_{=r}(x,y)$, etc.; see Appendix~\ref{sec:dist} for details.  Using these, we can further define FO-formulas $\inr_\range(x,y)$ for $\bar x - \bar y \in \range$, $\suc(x,y)$ for `$x$ is an immediate successor of $y$ in $\D$'\!, and FO-expressible constants $\min = 0$ and $\max = k$.

An \emph{event-based} \emph{interpretation} over $\D$ is a structure
$$
\I ~=~ (\varDelta, < , \varTheta , \mathsf{bit}_{\it in}, \mathsf{bit}_{\it fr}, A^\I_1, \dots, A^\I_p), \ \ \text{$A^\D_i \subseteq A^\I_i \subseteq \varTheta$,}
$$
where the Boolean connectives are interpreted as usual and
%
\begin{align*}
(\diamondmin_\range A)^\I &~=~ \{t \in \varTheta \mid \exists t' \in \varTheta \, (\inr_\range(t,t') \land t' \in A^\I)\},\\
(\boxminus_\range A)^\I &~=~ \{t \in \varTheta \mid \forall t' \in \varTheta \, (\inr_\range(t,t') \to t' \in A^\I)\}.
\end{align*}
An interpretation $\I$ over $\D$ is a \emph{model} of an $\MTL$-program $\Pi$ and $\D$ if, for any rule~\eqref{eq:rule} in $\Pi$ and any $t \in \varTheta$, whenever $t \in \vartheta_i^\I$ for all $i$, $1 \leq i \leq k$, then $t \in \vartheta_{k+j}^\I$ for some $j$, $1 \leq j \leq l$. We call $\D$ and $\Pi$ \emph{consistent} if there is a model of $\Pi$ and $\D$.

Henceforth, we write $\tem(\D)$ for the set $\varTheta$ of timestamps in~\eqref{struc} and often informally identify $t \in \tem(\D)$ with its value $\bar t$. We call $t \in \tem(\D)$ (and so $\bar t$) a \emph{certain answer} to $\q = (\Pi,A)$ over $\D$ if $t \in A^\I$ for every model $\I$ of $\D$ and $\Pi$.
The \emph{OMQ answering problem} for $\q$ is to decide, given $\D$ and $t \in \tem(\D)$, whether $t$ is a certain answer to $\q$ over $\D$.
%
To illustrate, consider \mbox{$\Pi = \{\boxminus_{[0,2)} B \to B', \ \diamondmin_{[1,1]} B' \to A\}$},  $\D_1 = \{B(0),B(1/2), C(3/2)\}$ and $\D_2 = \{B(0),C(3/2)\}$. Then $3/2$ is a certain  answer to $(\Pi,A)$ over $\D_1$, but there are no certain answers  to  $(\Pi,A)$ over $\D_2$:\\
\centerline{
\begin{tikzpicture}[nd/.style={draw,thick,circle,inner sep=0pt,minimum size=1.5mm,fill=white},xscale=1.1,>=latex]
\draw[thick,gray,->] (-0.5,0) -- (4,0);
\slin{0,0}{$B\ \textcolor{gray}{B'}$}{$0$};
\slin{1,0}{$B\  \textcolor{gray}{B'}$}{$\frac{1}{2}$};
\slin{3,0}{$C\ \textcolor{gray}{A}$}{$\frac{3}{2}$};
\end{tikzpicture}}
We are interested in the \emph{data complexity} of OMQ answering, that is,  regard $\D$ as the only input to the problem and assume $\q$ to be fixed.

Let $\mathcal{L}$ be a query language over FO-structures~\eqref{struc}. An OMQ $\q$ is said to be $\mathcal{L}$-\emph{rewritable} if there is an $\mathcal{L}$-query $\Q(x)$, called an \emph{$\mathcal{L}$-rewriting} of $\q$, such that, for any data instance $\D$, a timestamp $t\in \tem(\D)$ is a certain answer to $\q$ over $\D$ iff $\D \models \Q(t)$.
%
%
Our target query languages $\mathcal{L}$ include:
\begin{itemize}\itemsep=0pt
\item[--] FO$(<)$ and its extension FO$(<,+)$ with the predicate PLUS (e.g., $\exists x \,  \text{PLUS}(x, x, \max)$ says that $|\varTheta|$ is odd); evaluating such queries is in \ACz{} for data complexity;

\item[--] FO(RPR), i.e., FO$(<)$ with relational primitive recursion, which is in \NCo~\cite{DBLP:journals/iandc/ComptonL90};

\item[--] FO(TC) and FO(DTC), i.e., FO$(<)$ with transitive and deterministic transitive closure, which are in NL and L, respectively~\cite{Immerman99};

\item[--] datalog(FO), i.e., datalog queries with additional FO-formulas built from EDB predicates in their rule bodies, which are in \PTime~\cite{DBLP:conf/csl/Gradel91}.
\end{itemize}
All of them save datalog(FO) can be implemented in SQL.
$\mathcal{L}$-rewritability of an OMQ $\q$ means that answering $\q$ is in the same data-complexity class as evaluation of $\mathcal{L}$-queries.

%

Given a $\hMTL$-program $\Pi$ and a data instance $\D$, we define a set $\mathfrak C_{\Pi,\D}$ of pairs of the form $(\vartheta,t)$ that contains all answers to OMQs with $\Pi$ over $\D$.
We start by setting $\C = \D$ and denote by $\cl(\C)$ the result of applying exhaustively and non-recursively the following rules
to $\C$:
\begin{itemize}\itemsep=0pt
\item[--] if $\vartheta_1\land \dots \land \vartheta_k \to \vartheta$ is in $\Pi$ and $(\vartheta_i,t) \in \C$, for all $i$, $1 \leq i \leq k$, then we add $(\vartheta,t)$ to $\C$;

\item[--] if $\diamondmin_{\range}B$ occurs in $\Pi$,  $(B,t') \in \C$, and $\inr_\range(t,t')$ holds for some  $t \in \tem(\D)$, then we add $(\diamondmin_{\range} B,t)$ to $\C$;

\item[--] if $\boxminus_{\range}B$ occurs in $\Pi$, $t \in \tem(\D)$ and $(B,t') \in \C$ for all $t' \in \tem(\D)$ with $\inr_\range(t,t')$, then we add $(\boxminus_{\range} B,t)$ to $\C$.
%
\end{itemize}
It should be clear that there is some $N < \omega$ polynomially depending on $\Pi$ and $\D$ such that $\cl^N(\C) = \cl^{N+1}(\C)$. We then set $\mathfrak C_{\Pi,\D} = \cl^N(\D)$. We can regard $\mathfrak C_{\Pi,\D}$ as a (minimal) model of $\Pi$ and $\D$ with domain $\tem(\D)$ in which $t \in B^{\mathfrak C_{\Pi,\D}}$ iff $(B,t) \in \mathfrak C_{\Pi,\D}$ The proof of the following is standard:

\begin{theorem}\label{thm:canon} For a \hMTL-OMQ $(\Pi,A)$,
$(i)$ $\Pi$ is inconsistent with $\D$ iff $(\bot,t) \in \mathfrak C_{\Pi,\D}$; $(ii)$ a timestamp $t \in \tem(\D)$ is a certain answer to a \hMTL-OMQ $(\Pi,A)$ over $\D$ iff either $\mathfrak C_{\Pi,\D} \models A[t]$ or $\Pi$ is inconsistent with $\D$.
\end{theorem}

Note in passing that, as a consequence, we obtain the following reduction of $\mathcal{L}$-rewritability of more general \hMTL-OMQs $(\Pi,\varphi)$ with \emph{positive FO-queries} $\varphi$ (built from atoms, $\land$, $\lor$, $\forall$, and $\exists$) to $\mathcal{L}$-rewritability of atomic OMQs we deal with in this paper:
\begin{corollary}
Let $(\Pi,\varphi)$ be a \hMTL-OMQ with a positive FO-query $\varphi$. If $(\Pi,A)$ has an $\mathcal{L}$-rewriting $\Q_A(x)$, for every atom $A$, then $\Q_\varphi = \varphi[A_1/\Q_{A_1},\dots,A_n/\Q_{A_n}] \lor \exists x\, \Q_B(x)$ is an $\mathcal{L}$-rewriting of $(\Pi,\varphi)$, where $B$ is an atom not occurring in $\Pi$ and any data instance and $\varphi[A_1/\Q_{A_1},\dots,A_n/\Q_{A_n}]$ is the result of replacing every atom of the form $A_i(x)$ in $\varphi$ with $\Q_{A_i}(x)$.
\end{corollary}
\begin{proof}
Observe first that, since $\varphi(\vec{x})$ is positive, we have, for any consistent $\Pi$ and $\D$ and any $\vec{a} \subseteq \tem(\D)$, that $\mathfrak C_{\Pi,\D} \models \varphi(\vec{a})$ iff $\I \models \varphi(\vec{a})$ for all models $\I$ of $\Pi$ and $\D$. Indeed, to show $(\Rightarrow)$, we use the fact that there is a homomorphism from $\mathfrak C_{\Pi,\D}$ onto $\I$ (as $\mathfrak C_{\Pi,\D}$ is a minimal model of $\Pi$ and $\D$), and positive formulas are preserved under homomorphic images~\cite{Chang&Keisler}.
On the other hand, one can show by induction on the construction of $\varphi$ that $\mathfrak C_{\Pi,\D} \models \varphi(\vec{a})$ iff $\D \models \varphi[A_1/\Q_{A_1},\dots,A_n/\Q_{A_n}](\vec{a})$. It remains to observe that $\Pi$ and $\D$ are inconsistent iff $\D \models \exists x\, \Q_B(x)$.
\end{proof}
%



\section{OMQs with Arbitrary Ranges}\label{sec:arb}

We begin by establishing (non-)rewritability and data complexity of answering OMQs in various classes where  \emph{arbitrary} ranges in temporal operators are allowed. We denote by $\cMTL\!^\boxminus$  ($\cMTL\!^\diamondminus$) the restriction of $\cMTL$ to the language with operators $\boxminus_\range$ (respectively, $\diamondmin_\range$) only.

\begin{theorem}\label{start-complexity}
$(i)$ Answering \MTL-OMQs is \coNP-complete for data complexity;
$(ii)$ \hMTL-OMQs are \emph{datalog(FO)}-rewritable, with $\cMTL\!^\boxminus$-OMQs being \PTime-hard;
$(iii)$ $\cMTL\!^\diamondminus$-OMQs are \emph{FO(TC)}-rewritable and \NL-hard.
\end{theorem}
\begin{proof}
$(i)$ The membership in \coNP{} is trivial. We establish \coNP-hardness by reduction of \NP-complete circuit satisfiability~\cite{Arora&Barak09}.
Let $\boldsymbol{C}$ be a Boolean circuit with $N_0$-many (two-input) AND, OR and (one-input) NOT gates enumerated by consecutive numbers starting from 0 so that if there is an edge from $n$ to $m$, then $n < m$. Take the minimal $N = 2^k \ge N_0$ and a data instance $\D_{\boldsymbol{C}}$ \mbox{with the facts}
\begin{itemize}\itemsep=0pt
\item[--] $A(2n+ i/N)$, if $n$ is a gate and $0\leq i < N_0$;
\item[--] $X(2n+ n/N)$, if $n$ is an input  gate;
\item[--] $N(2n+n/N)$, if $n$ is a NOT gate;
\item[--] $D(2n+n/N)$, if $n$ is an OR gate;
\item[--] $C(2n+n/N)$, if $n$ is an AND gate;
\item[--] $I_0(2n + m/N)$, if $n$ is a NOT gate with input gate $m$;
\item[--] $I_1(2n + m/N)$ and $I_2(2n + k/N)$, if $n$ is an OR or AND gate with input gates $m$ and $k$.
\end{itemize}
Let $\Pi_{\boldsymbol{C}}$ be an $\MTL$-program with the following rules:
\begin{align*}
& X \to T \lor F,
\quad  \diamondmin_{[2,2]} T  \to T, \quad  \diamondmin_{[2,2]} F  \to F ,\\
&N \land \diamondmin_{[0,1]}(I_0 \land T)   \to   F, \quad  N \land \diamondmin_{[0,1]}(I_0 \land F)  \to T,\\
&D \land \diamondmin_{[0,1]}(I_1 \land T)   \to   T, \quad D \land \diamondmin_{[0,1]}(I_2 \land T)  \to T,\\
& C \land \diamondmin_{[0,1]}(I_1 \land F)   \to   F, \quad C \land \diamondmin_{[0,1]}(I_2 \land F)  \to F,
\\
& D \land \diamondmin_{[0,1]}(I_1 \land F) \land \diamondmin_{[0,1]}(I_2 \land F) \to
F,\\
& C \land \diamondmin_{[0,1]}(I_1 \land T) \land \diamondmin_{[0,1]}(I_2 \land T) \to
T.
\end{align*}
Then $\mathbf{C}$ is satisfiable iff the maximal number in $\tem(\D)$  is not a certain answer to $(\Pi_{\boldsymbol{C}},F)$ over $\D_{\boldsymbol{C}}$. An example of $\boldsymbol{C}$ and an initial part of a model of $\Pi_{\boldsymbol{C}}$, $\D_{\boldsymbol{C}}$ is shown below:
\begin{center}
\begin{tikzpicture}
\begin{scope}[scale = 1.3, xshift = 160, circuit logic US, every circuit symbol/.style={thick},thick]

\node[or gate,inputs={nn}, point right,fill=gray!20,label={[yshift = -.05cm]above left:{\scriptsize \smash{2}}}] (n2) at (0,1) {$\lor$};

\node[and gate,inputs={nn}, point right,fill=gray!20,label={[xshift=.25cm, yshift = -.05cm] above left:{\scriptsize \smash{3}}}] (n3) at (1.3,0.7) {$\land$};

\node[not gate,inputs={nn}, point right,fill=gray!20,label={[xshift=.25cm, yshift = .04cm] above left:{\scriptsize \smash{4}}}] (n4) at (2.3,0.7) {$\neg$};

\node[rectangle, draw, label=160:{\scriptsize \!1}] (x1) at (-1.23,0.5) {\footnotesize $X$};
\node[rectangle, draw, label=160:{\scriptsize \!0}] (x0) at (-1.23,1.1) {\footnotesize $X$};
\draw (n2.output) -- +(0.2,0) |- (n3.input 1);
\draw (n3.output) -- (n4.input);
\draw (n4.output) -- +(0.2,0);

\draw (n2.input 1) -- (x0);
\draw (n2.input 2) -| ($(x1)+(0.5,0)$);
\draw (x1) -- ($(x1)+(0.5,0)$);

\draw ($(n3.input 2)+(0,-.1)$) -- ($(x1)+(0.5,0)$);

\end{scope}

\begin{scope}[scale = 1.15, yshift = -12, xshift = 100, circuit logic US, every circuit symbol/.style={thick},thick]
\draw[thick,-, gray](0.1,0)--(7.2,0);
	
		
\foreach \x in {0,2,4,6}
{
\foreach \y in {1,...,5}
{
\draw[-, thick, black] (  $(\x + \y/4.5,0)+(0,0.05)$) -- ($(\x+ \y/4.5,0)+(0,-0.05)$ );
\node(\x and\y) at ($(\x+ \y/4.5,0)+(0,-0.3)$) { } ;
\node  at ($(\x+ \y/4.5,0)+(0,0.2)$) {\scriptsize$A$} ;
}
}

%

\node at ($(0and1)+(0,-0.05)$) {\scriptsize $\frac{0}{8}$} ;
\node at ($(0and2)+(0,-0.05)$) {\scriptsize $\frac{1}{8}$} ;
\node at ($(0and3)+(0,-0.05)$) {\scriptsize $\frac{2}{8}$} ;
\node at ($(0and4)+(0,-0.05)$) {\scriptsize $\frac{3}{8}$} ;
\node at ($(0and5)+(0,-0.05)$) {\scriptsize $\frac{4}{8}$} ;

\node at ($(2and1)+(0,-0.05)$) {\scriptsize $\frac{16}{8}$} ;
\node at ($(2and2)+(0,-0.05)$) {\scriptsize $\frac{17}{8}$} ;
\node at ($(2and3)+(0,-0.05)$) {\scriptsize $\frac{18}{8}$} ;
\node at ($(2and4)+(0,-0.05)$) {\scriptsize $\frac{19}{8}$} ;
\node at ($(2and5)+(0,-0.05)$) {\scriptsize $\frac{20}{8}$} ;

\node at ($(4and1)+(0,-0.05)$) {\scriptsize $\frac{32}{8}$} ;
\node at ($(4and2)+(0,-0.05)$) {\scriptsize $\frac{33}{8}$} ;
\node at ($(4and3)+(0,-0.05)$) {\scriptsize $\frac{34}{8}$} ;
\node at ($(4and4)+(0,-0.05)$) {\scriptsize $\frac{35}{8}$} ;
\node at ($(4and5)+(0,-0.05)$) {\scriptsize $\frac{36}{8}$} ;

\node at ($(6and1)+(0,-0.05)$) {\scriptsize $\frac{48}{8}$} ;
\node at ($(6and2)+(0,-0.05)$) {\scriptsize $\frac{49}{8}$} ;
\node at ($(6and3)+(0,-0.05)$) {\scriptsize $\frac{50}{8}$} ;
\node at ($(6and4)+(0,-0.05)$) {\scriptsize $\frac{51}{8}$} ;
\node at ($(6and5)+(0,-0.05)$) {\scriptsize $\frac{52}{8}$} ;

\node at  ($(0and1)+(0,0.7)$) {\scriptsize$X$} ;		
\node at  ($(2and2)+(0,0.7)$) {\scriptsize$X$} ;	
\node at  ($(4and1)+(0,0.7)$) {\scriptsize$I_1$} ;		
\node at  ($(4and2)+(0,0.7)$) {\scriptsize$I_2$} ;	
\node at  ($(4and3)+(0,0.7)$) {\scriptsize$D$} ;	
\node at  ($(6and3)+(0,0.7)$) {\scriptsize$I_1$} ;		
\node at  ($(6and2)+(0,0.7)$) {\scriptsize$I_2$} ;	
\node at  ($(6and4)+(0,0.7)$) {\scriptsize$C$} ;	

\node[gray] at  ($(0and1)+(0,0.9)$) {\scriptsize$T$} ;	
\node[gray] at  ($(2and1)+(0,0.9)$) {\scriptsize$T$} ;	
\node[gray] at  ($(2and2)+(0,0.9)$) {\scriptsize$F$} ;	
\node[gray] at  ($(4and1)+(0,0.9)$) {\scriptsize$T$} ;	
\node[gray] at  ($(4and2)+(0,0.9)$) {\scriptsize$F$} ;	
\node[gray] at  ($(4and3)+(0,0.9)$) {\scriptsize$T$} ;				
\node[gray] at  ($(6and1)+(0,0.9)$) {\scriptsize$T$} ;	
\node[gray] at  ($(6and2)+(0,0.9)$) {\scriptsize$F$} ;	
\node[gray] at  ($(6and3)+(0,0.9)$) {\scriptsize$T$} ;	
\node[gray] at  ($(6and4)+(0,0.9)$) {\scriptsize$F$} ;



\end{scope}
\end{tikzpicture}
\end{center}

$(ii)$ We construct a datalog(FO) rewriting $(\Pi',G(x))$ of a \hMTL-OMQ $\q = (\Pi,A)$. To begin with, we add to $\Pi$ the rule $P(x) \to P'(x,x)$ for each $P$ in $\Pi$. The other rules in $\Pi'$ are obtained from the rules in $\Pi$ by the following transformations. We replace every atom $B$ not under the scope of a temporal operator with $B'(x,x)$ and  every $\diamondmin_{[r,s]} B$ with
$$
B'(w,z) \land \mathsf{dist}_{\geq r}(x,w) \land \mathsf{dist}_{\leq s}(x,z)
$$
and similarly for other types of ranges $\range$ in $\diamondmin_\range B$. Intuitively, $\Pi', \D \models B'(x,y)$ iff $(B,t) \in \mathfrak C_{\Pi,\D}$, for each $t \in [x,y]$ from $\tem(\D)$. We replace every $\boxminus_{[r,s]} B$ in the body of a rule with
$$
  B'(w,z) \land \mathsf{dist}_{\geq s}(x,w) \land \mathsf{dist}_{\leq r}(x,z) \land
   \mathsf{dist}_{\ge (s-r)}(z,w)
$$
and similarly for other types of ranges.
Finally, we add the following rules to the resulting program:
\begin{align*}
& A'(y,z) \land (y \le x \le z) \to G(x),\\
& B'(x,y) \land B'(z,z) \land \mathsf{suc}(y,z) \to B'(x,z).
\end{align*}
Note that the obtained datalog program $\Pi'$ contains FO-definable EDB predicates such as $\mathsf{dist}_{\geq r}(x,w)$ and $\mathsf{suc}(y,z)$ in rule bodies. Clearly, $t$ is a certain answer to $\q$ over any given data instance $\D$ iff $t$ is an answer to $(\Pi', G(x))$ over $\D$.

We show \PTime{} hardness of $\cMTL\!^\boxminus$-OMQs by reduction of path system accessibility (PSA). Let  $G$  be a hypergraph
with $N_0$ vertices enumerated by consecutive natural numbers starting from 0 so that if $(m,n,o)$ is a hyperedge, then $m<n<o$.
Let $e_0,  \dots , e_{k-1}$ be  the lexicographical order of hyperedges.
Suppose the problem is to check whether a vertex $t$ is accessible from a set of vertices $S$,  i.e., whether $t \in S$ or there are vertices  $u,w$ accessible from $S$ and $(u,w,t)$ is a hyperedge. Let $\D_G$ comprise the atoms $A(4i +  n/N )$, for $0 \leq i \leq k$ and a vertex $n$, together with
\begin{itemize}
\item[--] $A(2+  4i +  m/N )$, $A(2+  4i + n/N )$, $A(2+  4i +  o/N )$, and $A(2+  4i + n/N  -1)$,
 for a hyperedge $e_i=(m,n,o)$;

\item[--] $R(4i + n/N )$, for $0 \leq i \leq m$ and  $n \in S$.
\end{itemize}
For example, for the vertices $ 0,1,2,3$, hyperedge $(0,1,2)$, $S=\{ 0,1 \}$, and $t=3$, $\D_G$ looks as follows:
\begin{center}
\begin{tikzpicture}
[scale = 1.3,>=stealth']
	
\tikzset{>=latex}
		
\draw[->, gray](-0.3,0)--(5.3,0);
	
		
\foreach \x in {0,4}
{
\foreach \y in {1,...,4}
{
\draw[-, thick, black] (  $(\x + \y/5,0)+(0,0.05)$) -- ($(\x+ \y/5,0)+(0,-0.05)$ );
\node(\x and\y) at ($(\x+ \y/5,0)+(0,-0.3)$) {} ;

\node at ($(\x and\y)+(0,0.5)$) {\scriptsize $A$} ;
}
}

\node(aux) at ($(1.4,0)+(0,-0.3)$) {} ;
\node at ($(aux)+(0,0.5)$) {\scriptsize $A$} ;
\draw[-, thick, black] (1.4,0.05)-- (1.4,-0.05);

\node at ($(aux)+(0,-0.05)$) {\scriptsize $\frac{5}{4}$} ;
\node(2and1) at ($(2.2,0)+(0,-0.3)$) {} ;
\node at ($(2and1)+(0,0.5)$) {\scriptsize $A$} ;
\draw[-, thick, black] (2.2,0.05)-- (2.2,-0.05);

\node(2and2) at ($(2.4,0)+(0,-0.3)$) {} ;
\node at ($(2and2)+(0,0.5)$) {\scriptsize $A$} ;
\draw[-, thick, black] (2.4,0.05)-- (2.4,-0.05);

\node(2and3) at ($(2.6,0)+(0,-0.3)$) {} ;
\node at ($(2and3)+(0,0.5)$) {\scriptsize $A$} ;
\draw[-, thick, black] (2.6,0.05)-- (2.6,-0.05);

\node at ($(0and1)+(0,-0.05)$) {\scriptsize $\frac{0}{4}$} ;
\node at ($(0and2)+(0,-0.05)$) {\scriptsize $\frac{1}{4}$} ;
\node at ($(0and3)+(0,-0.05)$) {\scriptsize $\frac{2}{4}$} ;
\node at ($(0and4)+(0,-0.05)$) {\scriptsize $\frac{3}{4}$} ;

\node at ($(2and1)+(0,-0.05)$) {\scriptsize $\frac{8}{4}$} ;
\node at ($(2and2)+(0,-0.05)$) {\scriptsize $\frac{9}{4}$} ;
\node at ($(2and3)+(0,-0.05)$) {\scriptsize $\frac{10}{4}$} ;

\node at ($(4and1)+(0,-0.05)$) {\scriptsize $\frac{16}{4}$} ;
\node at ($(4and2)+(0,-0.05)$) {\scriptsize $\frac{17}{4}$} ;
\node at ($(4and3)+(0,-0.05)$) {\scriptsize $\frac{18}{4}$} ;
\node at ($(4and4)+(0,-0.05)$) {\scriptsize $\frac{19}{4}$} ;
		
\node at  ($(0and1)+(0,0.7)$) {\scriptsize$R$} ;		
\node at  ($(0and2)+(0,0.7)$) {\scriptsize$R$} ;		

\node at  ($(4and1)+(0,0.7)$) {\scriptsize$R$} ;		
\node at  ($(4and2)+(0,0.7)$) {\scriptsize$R$} ;		
	
			
\node[gray] at  ($(2and1)+(0,1)$) {\scriptsize$R'$} ;
\node[gray] at  ($(2and2)+(0,1)$) {\scriptsize$R'$} ;
\node[gray] at  ($(2and3)+(0.05,1)$) {\scriptsize$R''$} ;	

\node[gray] at  ($(4and3)+(0,1)$) {\scriptsize$R$} ;			
		
\draw [gray,draw=none, decorate,decoration={brace,amplitude=4pt,raise=0pt},yshift=0pt]
(-0.1,0.7) -- (-0.1,0.7) node [black,midway,xshift=-0.6cm] {\color{gray} \footnotesize by~$\Pi$:};

\draw [decorate,draw=none,decoration={brace,amplitude=4pt,raise=0pt},yshift=0pt]
(-0.1,0.02) -- (-0.1,0.6) node [black,midway,xshift=-0.5cm] {\footnotesize
$\mathcal{D}_G:$};
		
\draw [decorate,decoration={brace,mirror,amplitude=4pt,raise=0pt},yshift=0pt]
(1.25,-0.6) -- (2.65,-0.6) node [black,midway,yshift=-0.4cm] {\footnotesize hyperedge $(v_0,v_1,v_2)$
};	
\end{tikzpicture}
\end{center}
Let $\Pi$ be a $\cMTL\!^\boxminus$ program with the rules:
$$
\boxminus_{[2,2]}    R  \to R'\!,
\boxminus_{(0,1]}   R' \to R''\!\!,
\boxminus_{[2,2]} R''\!\to R,
\boxminus_{[4,4]} R \to R.
$$
Then $4k +  t/N$ is a certain answer to $(\Pi,R)$ over $\mathcal{D}_G$ iff $t$ is accessible from $S$ in $G$.

$(iii)$
%
The upper bound can be shown by reduction to FO(TC) via linear datalog(FO). Without loss of generality, we assume that, in the disjointness constraints $\vartheta_1 \land \vartheta_2 \to \bot$ occurring in the given $\cMTL\!^\diamondminus$-OMQ $\q = (\Pi,A)$, the $\vartheta_i$ are atomic. First, we straightforwardly translate $\q$ with the disjointness constraints removed from $\Pi$ to linear datalog(FO). Then, we transform the result into an FO(TC)-query $\varPsi_A(x)$~\cite{DBLP:conf/csl/Gradel91}. Now, for every disjointness constraint $B_1 \land B_2 \to \bot$ in $\Pi$, we take the sentence $\exists x (\varPsi_{B_1}(x) \land \varPsi_{B_2}(x))$ and, finally, form a disjunction of $\varPsi_A(x)$ with those sentences, which is obviously an FO(TC)-rewriting of $\q$.

We prove \NL-hardness by reduction of the reachability problem in acyclic digraphs. Let $G$ be such a digraph with $N_0$ vertices enumerated by consecutive natural numbers starting from 0 so that, if there is an edge from $n$ to $m$, then $n<m$. Let $e_0,  \dots , e_{k-1}$ be  the lexicographical order of edges. 
Take the minimal $N = 2^i \ge N_0$ for $i \in \mathbb{N}$. Suppose we want to check whether a vertex $t$ is accessible from $s$. Let $\mathcal{D}_G$ consist of the atoms
%
$A(4i +  n/N )$, for $0 \leq i \leq k$ and a vertex $n$;		
%
$A(2+  4i +  n/N  )$, $A(2+  4i + m/N )$,   for every edge $e_i=( n,m)$;		
%
$R(4i +  s/N )$, for  $0 \leq i \leq k$.
%
%
An example of $G$ and an initial part of $\D_G$ is shown below:
\begin{center}
\begin{tikzpicture}
[scale = 1.3,>=stealth']
	
\scriptsize	

\tikzset{>=latex}
				
\node[ label={[yshift=-15,xshift=-8]{$s=0$}}](v0) at (0.8,1.3) {\textbullet} ;
\node[ label={[yshift=-15]{$1$}}](v1) at (2.8,1.5) {\textbullet} ;
\node[ label={[yshift=-15]{$2$}}](v2) at (2.8,1) {\textbullet} ;
\node[ label={[yshift=-15,xshift=5]{$3=t$}}](v3) at (4.8,1.3) {\textbullet} ;
		
\path[->]
(v0) edge[bend left=-7]  node {} (v2)
(v2) edge[bend left=-7] node {} (v3)
(v1) edge[bend left=7]  node {} (v3);
		
\draw[thick,->, gray](-0.3,0)--(5.3,0);
	
		
\foreach \x in {0,4}
{
\foreach \y in {1,...,4}
{
\draw[-, thick, black] (  $(\x + \y/5,0)+(0,0.05)$) -- ($(\x+ \y/5,0)+(0,-0.05)$ );
\node(\x and\y) at ($(\x+ \y/5,0)+(0,-0.3)$) {} ;
						
\node at ($(\x and\y)+(0,0.5)$) {\scriptsize $A$} ;

}
}

\node(2and1) at ($(2.2,0)+(0,-0.3)$) {} ;
\node at ($(2and1)+(0,0.5)$) {\scriptsize $A$} ;
\draw[-, thick, black] (2.2,0.05)-- (2.2,-0.05);

\node(2and3) at ($(2.6,0)+(0,-0.3)$) {} ;
\node at ($(2and3)+(0,0.5)$) {\scriptsize $A$} ;
\draw[-, thick, black] (2.6,0.05)-- (2.6,-0.05);

\node at ($(0and1)+(0,0.07)$) {\scriptsize $\frac{0}{4}$} ;
\node at ($(0and2)+(0,0.07)$) {\scriptsize $\frac{1}{4}$} ;
\node at ($(0and3)+(0,0.07)$) {\scriptsize $\frac{2}{4}$} ;
\node at ($(0and4)+(0,0.07)$) {\scriptsize $\frac{3}{4}$} ;

\node at ($(2and1)+(0,0.07)$) {\scriptsize $\frac{8}{4}$} ;
\node at ($(2and3)+(0,0.07)$) {\scriptsize $\frac{10}{4}$} ;

\node at ($(4and1)+(0,0.07)$) {\scriptsize $\frac{16}{4}$} ;
\node at ($(4and2)+(0,0.07)$) {\scriptsize $\frac{17}{4}$} ;
\node at ($(4and3)+(0,0.07)$) {\scriptsize $\frac{18}{4}$} ;
\node at ($(4and4)+(0,0.07)$) {\scriptsize $\frac{19}{4}$} ;
		
\node at  ($(0and1)+(0,0.7)$) {\scriptsize$R$} ;		

\node at  ($(4and1)+(0,0.7)$) {\scriptsize$R$} ;		

			
\node[gray] at  ($(2and1)+(0,0.9)$) {\scriptsize$R'$} ;
\node[gray] at  ($(2and3)+(0,0.9)$) {\scriptsize$R''$} ;	

\node[gray] at  ($(4and3)+(0,0.9)$) {\scriptsize$R$} ;	

\draw [draw=none,decorate,decoration={brace,amplitude=4pt,raise=0pt},yshift=0pt]
(-0.1,1.2) -- (-0.1,1.4) node [black,midway,xshift=-0.4cm] {
$G$:};

\draw [gray,draw=none, decorate,decoration={brace,amplitude=4pt,raise=0pt},yshift=0pt]
(-0.1,0.4) -- (-0.1,0.9) node [black,midway,xshift=-0.6cm] {\color{gray}  by~$\Pi$:};

\draw [decorate,draw=none,decoration={brace,amplitude=4pt,raise=0pt},yshift=0pt]
(-0.1,0.02) -- (-0.1,0.6) node [black,midway,xshift=-0.5cm] {
$\mathcal{D}_G:$};

\draw [decorate,decoration={brace,mirror,amplitude=4pt,raise=0pt},yshift=0pt]
(2.1,-0.35) -- (2.7,-0.35) node [black,midway,yshift=-0.3cm] { edge $e_0=(0,2)$
};	

\end{tikzpicture}
\end{center}
Let $\Pi$ be a $\cMTL\!^\diamondminus$ program with the following rules:
%
$$
\diamondmin_{[2,2]} R \to R'\!, \,
\diamondmin_{(0,1]}   R' \!\to R''\!\! , \,
 \diamondmin_{[2,2]} R'' \!\to R,
\,
\diamondmin_{[4,4]} R \to R.
$$
Then $4k +  t/N$ is a certain answer to $(\Pi,R)$ over $\mathcal{D}_G$ iff $t$ is reachable from $s$ in $G$.
%
\end{proof}


To obtain finer complexity results, we classify \MTL-OMQs by the type of ranges $\range$ in their operators $\diamondmin_\range$ and $\boxminus_\range$: infinite, punctual, and non-punctual. Let $\langle$ be one of $($ or $[$, and let $\rangle$ be one of $)$ or $]$.


\section{OMQs with Ranges $\langle r,\infty)$}

First, consider OMQs with $\diamondmin_{\langle r,\infty)}$ and $\boxminus_{\langle r,\infty)}$, which resemble \LTL-operators `sometime' and `always in the past'\!. Using partially-ordered automata, it was shown in~\cite{DBLP:conf/ijcai/ArtaleKKRWZ15} that \LTL-OMQs with these operators are FO-rewritable. Although such automata are not applicable now, we establish the same complexity by characterising the structure of models. In the constructions below, it will be convenient to regard $\boxminus_\range$ as an abbreviation for $\neg \diamondmin_\range \neg$ with Boolean negation $\neg$ and only consider, without loss of generality, OMQs $(\Pi,A)$ with $A$ occurring in $\Pi$.
\begin{theorem}\label{aco}
\MTL-OMQs with temporal operators of the form $\diamondmin_{\langle r, \infty)}$ and $\boxminus_{\langle r, \infty)}$ only are $\textup{FO}(<)$-rewritable.
\end{theorem}
\begin{proof}
Let $\q = (\Pi,A)$ be an \MTL-OMQ as specified above. A \emph{simple literal}, $\slit$, for $\Pi$ takes the form $P$ or $\neg P$, where $P$ is an atom in $\Pi$; a \emph{temporal literal}, $\tlit$, for $\Pi$ is of the form $\diamondmin_\range \sigma$ or $\neg  \diamondmin_\range \sigma$ provided that $\diamondmin_\range P$ or $\boxminus_\range P$ occurs in $\Pi$ and $P$ is the atom in $\sigma$. Let $\varSigma_\Pi$ and $\varXi_\Pi$ be the sets of simple and temporal literals for $\Pi$, respectively.
A \emph{type} for $\Pi$ is any maximal set $\type \subseteq \varSigma_\Pi \cup \varXi_\Pi$  consistent with $\Pi$. The number of different types is $N_\Pi = 2^{O(|\Pi|)}$.

Given a model $\I$ of $\Pi$ and some $\D$ with $s\in\tem(\D)$, denote by $\type(s)$ the type of $s$ in $\I$. As the ranges in $\Pi$ are of the form $\langle r, \infty)$, the model $\I$ has the following \emph{monotonicity property}:
\begin{itemize}
\item[--] $\diamondmin_\range \slit \in \type(s)$ implies $\diamondmin_\range \slit \in \type(s')$ for all $s' > s$ in $\I$;

\item[--] $\neg \diamondmin_\range \slit \in \type(s)$ implies $\neg \diamondmin_\range \slit \in \type(s')$ for all $s' < s$ in $\I$.
\end{itemize}
We call $\type(s)$ in $\I$ an \emph{osteo-type} if there is $\lit \in \type(s)$ such that $\lit \notin \type(s')$, for all $s' < s$. Thus, if $\diamondmin_\range \slit \in \type(s')$ in $\I$, there is an osteo-type $\type(s) \ni \slit$ with $\inr_\range(s',s)$.  All osteo-types in $\I$ are pairwise distinct, so the number of them does not exceed $N_\Pi$.
Non-osteo-types are called \emph{fluff-types}. By monotonicity, any fluff-type $\type(s')$ has the same temporal literals as its nearest osteo-type $\type(s)$, for $s < s'$. For example, in the model of the program $\Pi = \{\boxminus_\range P \land \diamondmin_\range P \land P \to \bot\}$, $\range = [1,\infty)$, shown below, there are three fluff-types: $\type(3/4)$, $\type(9/8)$, and $\type(5/4)$.
\begin{center}
\begin{tikzpicture}
[scale = 1.3,>=stealth']
	
\tikzset{>=latex}
		
\draw[thick,->, gray](-0.2,0)--(6.2,0);

\scriptsize

\node(rho1) at (0.2,0){ };
\node(rho2) at (1.6,0){ };
\node(rho3) at (2.4,0){ };
\node(rho4) at (3.3,0){ };
\node(rho5) at (4.1,0){ };
\node(rho6) at (4.8,0){ };
\node(rho7) at (5.8,0){ };

\node(rho1_up)  at ($(rho1)+(0,0.06)$) {};
\node(rho1_down)  at ($(rho1)+(0,-0.06)$) {};
\draw[-, very thick] (  rho1 |- rho1_up) -- (rho1 |- rho1_down);
\node[below  = 0.01 of rho1] {$0$};
\node[above  = 0 of rho1] {$\neg \diamondmin_\range P$};
\node[above  = 0.4 of rho1] {$\neg \diamondmin_\range \neg P$};
\node[above  = 0.8 of rho1] {$  \neg P$};

\node(rho2_up)  at ($(rho2)+(0,0.06)$) {};
\node(rho2_down)  at ($(rho2)+(0,-0.06)$) {};
\draw[-, very thick] (  rho2 |- rho2_up) -- (rho2 |- rho2_down);
\node[below  = 0.01 of rho2] {$ \frac{1}{2}$};
\node[above  = 0 of rho2] {$\neg \diamondmin_\range P$};
\node[above  = 0.4 of rho2] {$\neg \diamondmin_\range \neg P$};
\node[above  = 0.8 of rho2] {$   P$};

\node(rho3_up)  at ($(rho3)+(0,0.06)$) {};
\node(rho3_down)  at ($(rho3)+(0,-0.06)$) {};
\draw[-, very thick] (  rho3 |- rho3_up) -- (rho3 |- rho3_down);
\node[below  = 0.01 of rho3] {$\frac{3}{4}$};
\node[above  = 0 of rho3] {$\neg \diamondmin_\range P$};
\node[above  = 0.4 of rho3] {$\neg \diamondmin_\range \neg P$};
\node[above  = 0.8 of rho3] {$  \neg P$};

\node(rho4_up)  at ($(rho4)+(0,0.06)$) {};
\node(rho4_down)  at ($(rho4)+(0,-0.06)$) {};
\draw[-, very thick] (  rho4 |- rho4_up) -- (rho4 |- rho4_down);
\node[below  = 0.01 of rho4] {$1$};
\node[above  = 0 of rho4] {$\neg \diamondmin_\range P$};
\node[above  = 0.4 of rho4] {$ \diamondmin_\range \neg P$};
\node[above  = 0.8 of rho4] {$  \neg P$};

\node(rho5_up)  at ($(rho5)+(0,0.06)$) {};
\node(rho5_down)  at ($(rho5)+(0,-0.06)$) {};
\draw[-, very thick] (  rho5 |- rho5_up) -- (rho5 |- rho5_down);
\node[below  = 0.01 of rho5] {$\frac{9}{8}$};
\node[above  = 0 of rho5] {$\neg \diamondmin_\range P$};
\node[above  = 0.4 of rho5] {$ \diamondmin_\range \neg P$};
\node[above  = 0.8 of rho5] {$   P$};

\node(rho6_up)  at ($(rho6)+(0,0.06)$) {};
\node(rho6_down)  at ($(rho6)+(0,-0.06)$) {};
\draw[-, very thick] (  rho6 |- rho6_up) -- (rho6 |- rho6_down);
\node[below  = 0.01 of rho6] {$\frac{5}{4}$};
\node[above  = 0 of rho6] {$\neg \diamondmin_\range P$};
\node[above  = 0.4 of rho6] {$ \diamondmin_\range \neg P$};
\node[above  = 0.8 of rho6] {$   P$};

\node(rho7_up)  at ($(rho7)+(0,0.06)$) {};
\node(rho7_down)  at ($(rho7)+(0,-0.06)$) {};
\draw[-, very thick] (  rho7 |- rho7_up) -- (rho7 |- rho7_down);
\node[below  = 0.01 of rho7] {$\frac{3}{2}$};
\node[above  = 0 of rho7] {$ \diamondmin_\range P$};
\node[above  = 0.4 of rho7] {$ \diamondmin_\range \neg P$};
\node[above  = 0.8 of rho7] {$\neg   P$};

\draw [decorate,decoration={brace,mirror,amplitude=4pt,raise=0pt},yshift=0pt]
(3.9,-0.35) -- (5.1,-0.35) node [black,midway,yshift=-0.3cm] {fluff-types
};	

\draw [decorate,decoration={brace,mirror,amplitude=4pt,raise=0pt},yshift=0pt]
(2.1,-0.35) -- (2.7,-0.35) node [black,midway,yshift=-0.3cm] {fluff-type
};	
\end{tikzpicture}
\end{center}
We now define an FO-sentence $\Phi_\Pi$ such that any given data instance $\D$ is consistent with $\Pi$ iff $\Phi_\Pi$ holds in the FO-structure $\D$. Let $\mathfrak O_\Pi$ be the set of sequences $\bar\type=(\type_1,\dots,\type_n)$, $1 \le n \le N_\Pi$, of distinct types for $\Pi$ that satisfy the monotonicity property and such that $\diamondmin_\range \slit \in \type_i$ implies $\slit \in \type_j$ for some $j\le i$; for minimal such $j$, we write $\textit{wit}(\type_i,\type_j,\range)$. We write $\overline{\textit{wit}}(\type_i,\type_j,\range)$ if $j\le i$, $\neg\diamondmin_\range \slit \in \type(s_i)$ and $\slit \in \type(s_j)$, for some $\diamondmin_\range \slit$.
Denote by $\mathfrak F^i_{\bar\type}$ the set of types $\type$ for $\Pi$ sharing the same temporal literals with $\type_i$ and such that, for every $\sigma \in \type$, there is $\type_j \ni \sigma$ with $j \le i$.
%
Finally, for any type $\type$, let $\delta_\type(x) = \bigwedge_{\neg P \in \type} \neg P(x)$ (which is true at $t$ in $\D$ iff, for every $P$ in $\Pi$, whenever $P(t)  \in \D$ then $P(t) \in\type$). Now, we set
%
\begin{multline*}
\Phi_\Pi ~=~ \hspace*{-1mm} \bigvee_{\bar \type \in \mathfrak O_\Pi} \hspace*{-1mm}  \exists x_1,\dots,x_n \big[ (x_1 = \min)  \land \bigwedge_{1\le i\le n} \delta_{\type_i}(x_i) \land{}\\
\bigwedge_{{\it wit}(\type_i,\type_j,\range)} \inr_\range(x_i,x_j) \land  \bigwedge_{\overline{\it wit}(\type_i,\type_j,\range)} \neg\inr_\range(x_i,x_j)  \land{} \\
\forall y  \hspace*{-1mm} \bigwedge_{1\le i \le n}  \hspace*{-2mm} \big( (x_i \prec y) \to \bigvee_{\type \in \mathfrak F^i_{\bar\type}} (\delta_\type(y) \land{}   \hspace*{-2mm} \bigwedge_{\overline{\it wit}(\type_i,\type_j,\range)} \hspace*{-5mm} \neg\inr_\range(y,x_j) )\big)\big],
\end{multline*}
where $x_i \prec y$ says that $x_i$ is the nearest predecessor of $y$, which is different from $x_1,\dots,x_n$.

Suppose $\I$ is a model of $\Pi$, $\D$ and $\bar \type = (\type(t_1),\dots,\type(t_n))$, for $t_1 < \dots < t_n$, are all the osteo-types in $\I$. This $n$-tuple of types is in $\mathfrak O_\Pi$ and the $\delta_{\type(t_i)}(t_i)$ are true in $\I$ by definition. The $\inr_\range(t_i,t_j)$ also hold for ${\it wit}(\type(t_i),\type(t_j),\range)$ because $\type(t_i)$ is the first type in $\I$ witnessing the relevant $\diamondmin_\range \sigma$. Similarly, $\inr_\range(t_i,t_j)$ does not hold in $\I$ for $\overline{\it wit}(\type(t_i),\type(t_j),\range)$. Finally, let $t$ be any timestamp in $\I$ with $t_i \prec t$. By construction, $\type(t)$ is a fluff-type in $\mathfrak F^i_{\bar\type}$ and $\delta_{\type(t)}(t)$ holds in $\I$. If $\overline{\it wit}(\type(t_i),\type(t_j),\range)$, we have $\neg \diamondmin_{\range} \sigma \in \type(t_i) \cap \type(t)$ and $\sigma \in \type(t_j)$, and so  $\inr_{\range}(t,t_j)$ cannot hold in $\I$. Thus, $\D\models\Phi_\Pi$.

Conversely, suppose $\Phi_\Pi$ holds in $\D$, assigning timestamps $t_i$ to the $x_i$ and associating types $\type(t)$ with every $t \in \tem(\D)$. Define an interpretation $\I$ by setting
$$
P^\I = \{ t \in \tem(\D) \mid P \in \type(t)\},
$$
for every atom $P$. We prove that $\I$ is a model of $\Pi$ and $\D$. As all the $\type(t)$ are types for $\q$, it suffices to show that
$$
\diamondmin_\range \sigma \in \type(t) \quad \Longleftrightarrow \quad \exists t' \, \big( \inr_\range(t,t') ~\land~ \sigma \in \type(t')\big).
$$
Suppose $\diamondmin_\range \sigma \in \type(t)$. If $t=t_i$, for some $i$, then $\textit{wit}(\type_i,\type_j,\range)$, for some $j \le i$, and so $\sigma \in \type(t_j)$ and $\inr_\range(t_i,t_j)$ holds in $\I$. If $t_i \prec t$, then $\diamondmin_\range \sigma \in \type(t_i)$, and we can use the previous  argument as $\inr_\range(t_i,t_j)$ implies $\inr_\range(t,t_j)$.

Conversely, suppose $\diamondmin_\range \sigma \notin \type(t)$. Then $\neg\diamondmin_\range \sigma \in \type(t)$. Consider first the case $t = t_i$. Suppose $t' \le t_i$ with $\sigma \in \type(t')$. Then $\sigma \in \type(t_j)$ for some $t_j \le t'$, and so $\overline{\it wit}(\type_i,\type_j,\range)$ and $\neg\inr_\range(t_i,t_j)$, whence $\neg\inr_\range(t,t_j)$ and $\neg\inr_\range(t,t')$. Now, let $t \notin \{t_1,\dots, t_n\}$. Then $t_i \prec t$ for some $i$ (because of $\min x_1$). Suppose $t' \le t$ with $\sigma \in \type(t')$. If $t' < t_i$, then $\sigma \in \type(t_j)$ for some $t_j \le t'$, and so $\overline{\it wit}(\type_i,\type_j,\range)$ and $\neg\inr_\range(t,t_j)$, whence $\neg\inr_\range(t,t')$. If $t' = t_i$, then, by the last conjunct of $\Phi_\Pi$, we have $\neg\inr_\range(t,t_i)$. Finally, if $t_i < t' \le t$, then $\sigma \in \type(t_j)$, for some $t_j \le t_i$, and we are done again.

An FO$(<)$-rewriting of $\q$ is the FO formula $\neg \Phi_{\neg A}(x)$, where $\Phi_{\neg A}(x)$ is obtained from $\Phi_\Pi$ by replacing $\delta_\type(z)$ with $\delta_\type(z, x)$, which is $\delta_\type(z)$ if $\neg A \in \type$ and $\delta_\type(z) \land (x \neq z)$ otherwise. Clearly, $\Phi_{\neg A}(x)$ holds in $\D$ iff there is a model of $\Pi$ and $\D$ satisfying $\neg A$ in $x$.
%
\end{proof}

We also mention in passing one more FO-rewritability result (which does not fit our classification). To formulate it, we require a few definitions.

\textbf{Normal form.} Until the end of this section, we assume that our \MTL{} programs and OMQs are in normal form. Namely, a program is said to be in \emph{normal form} if its rules have one of the forms:
\begin{align}
& \diamondminus_{\range_1'} P'_1 \land \dots \land \diamondminus_{\range_\ell'} P'_m \to P_0, \label{eq:init}\\
&\diamondminus_{\range_1} P_1 \land \dots \land \diamondminus_{\range_k} P_k \land \diamondminus_{\range_1'} P'_1 \land \dots \land \diamondminus_{\range_\ell'} P'_m \to P_0, \label{eq:trans}
\end{align}
where the $P'_i$ are from the data alphabet (like EDB predicates in datalog) and do not occur in the rule heads, while the $P_i$ do not occur in data instances, and $0 \notin \range_i$ for any $i$ (although there may be $\range'_i=[0,0]$). Every $\hMTL\!^\diamondminus$ program can be transformed to a program in normal form with the same answers. We illustrate this claim by an example.
%
\begin{example}\em
Let $\Pi = \{  \diamondminus_{[0,d]} P_0' \land Q_0' \to P_1',\
\diamondminus_{(0,e)} P_1' \land \diamondminus_{[0,f]} Q_1' \to P_0'\}$,
where the $P_i'$ are in the data alphaet. By introducing fresh atoms $P_0$, $P_1$, we convert $\Pi$ to
$$
\diamondminus_{[0,d]} P_0 \land Q_0' \to P_1,\
\diamondminus_{(0,e)} P_1 \land \diamondminus_{[0,f]} Q_1' \to P_0, \
 P_0' \to P_0,\ P_1' \to P_1.
$$
To get rid of $[0,d]$, we further transform the program to
$$
P_0 \land Q_0' \to P_1,\
\diamondminus_{(0,d]} P_0 \land Q_0' \to P_1,\
\diamondminus_{(0,e)} P_1 \land \diamondminus_{[0,f]} Q_1' \to P_0, \
P_0' \to P_0,\ P_1' \to P_1.
$$
Now, $P_0$ in the first rule is not in the scope of $\diamondminus_{\range}$ ($Q_0'$ can be regarded as a shorthand for $\diamondminus_{[0,0]} Q_0'$). So we transform the rule  using obvious derivations to obtain the following program in normal form:
\begin{multline*}
P_0' \land Q_0' \to P_1,\ \diamondminus_{(0,e)} P_1 \land \diamondminus_{[0,f]} Q_1' \land Q_0' \to P_1,\
\diamondminus_{(0,d)} P_0 \land Q_0' \to P_1,\\
\diamondminus_{(0,e)} P_1 \land \diamondminus_{[0,f]} Q_1' \to P_0, \
P_0' \to P_0,\ P_1' \to P_1.
\end{multline*}
\end{example}

A $\hMTL\!^\diamondminus$ query $(\Pi, A(x))$ is in \emph{normal form} if $\Pi$ is in normal form and $A$ is not in the data alphabet. Clearly, every query can be converted to a one in normal form and having the same answers.

\textbf{Metric automata for $\hMTL\!^\diamondminus$.} Our technical tool for studying the data complexity of linear $\hMTL\!^\diamondminus$ queries is  automata with metric constraints that are defined for programs in normal form. These automata can be viewed as a primitive version of standard timed automata for \MTL~\cite{ALUR1994183} as we  only have one clock $c$, the clock reset $c:=0$ happens at every  transition, and the clock constraints are of the simple form $c \in \range$.
	
A (nondeterministic) \emph{metric automaton} is a quadruple $\A = (S, S_0, \Sigma, \delta)$, where $S \ne \emptyset$ is a set of \emph{states}, $\Sigma$ a \emph{tape alphabet}, $\delta$ a \emph{transition relation}, and $S_0$ is a nonempty set of pairs of the form $(q, e)$, where $e \in \Sigma$, $q \in S$. The transition relation $\delta$ is a set of instructions of the form $q \xrightarrow{\range}_e q'$ with $q,q' \in S$, $e \in \Sigma$ and a range $\range$. The automaton $\A$ takes as input \emph{timed words} $\sigma = (e_0, t_0), \dots, (e_n, t_n)$, where the $t_i$ are timestamps with $t_{i-1} < t_{i}$. A \emph{run} over $\sigma$ is a sequence $q_0, \dots, q_m$ such that $(q_0, e_0) \in S_0$, $q_{i-1} \xrightarrow{\range_i}_{e_i} q_{i}$ is in $\delta$ and $t_{i} - t_{i-1} \in \range_i$, for $0 < i \leq n$.
	
Let $\Pi$ be a \emph{linear} $\hMTL\!^\diamondminus$ program in normal form. We denote the conjunctions $\diamondminus_{\range_1'} P'_1 \land \ldots \land \diamondminus_{\range_\ell'} P'_m$ (with data atoms $P_i'$) that occur in $\Pi$ by $\varepsilon$, possibly with subscripts. Thus, since $\Pi$ is linear, rules~\eqref{eq:trans} in $\Pi$ are of the form $\varepsilon \land \diamondminus_{\range} Q \to P$. Let $E_\Pi = \{ \varepsilon_1, \dots, \varepsilon_q \}$ be the set of all such $\varepsilon$ occurring in $\Pi$.
We define a metric automaton $\mathcal{A}_\Pi$ for $\Pi$ as follows.  The set $S$ of its states comprises the head concept names in $\Pi$, and $\Sigma = 2^{E_\Pi}$.
The transition relation $\delta$ comprises $Q \xrightarrow{\range}_E P$ such that $\varepsilon \land \diamondminus_{\range} Q \to P$ is in $\Pi$ and $\varepsilon \in E$. Finally, $S_0$ is the set of all pairs $(P, \varepsilon)$ such that a rule $\varepsilon \to P$ of the form~\eqref{eq:init} is in $\Pi$.
\begin{example}\label{ex:automaton}\em
For $\Pi = \{ \diamondminus_{[0,1]} P_0' \to P_0,\
 \diamondminus_{(1,2)} P_0 \land P_1' \to P_1,\
\diamondminus_{(1,3)} P_1 \to P_0 \}$,
the metric automaton $\A_\Pi$ is depicted below, where $P_0',P_1' \in \Lambda$, $E_0 = \{P_1'\}$, $E_1 = \{\diamondminus_{[0,1]} P_0'\}$, $E_2 = \{P_1', \diamondminus_{[0,1]} P_0'\}$, and $S_0 =\{ (P_0, \diamondminus_{[0,1]} P_0') \}$.\\[3pt]
\centerline{\begin{tikzpicture}[>=latex,xscale=2.5,yscale=1,thick,
state/.style={circle,draw,fill=gray!50,thick,minimum size=4mm,inner sep=1pt}]\small
\node[state] (D) at (0,0) {$P_0$};
\node[state] (E) at (2,0) {$P_1$};
\draw[->] (D) to [out=0,in=180] node[midway, above] {$E_0\ (1,2)$} (E);
\draw[->] (D) to [out=30,in=150, looseness=1.5] node[midway, above] {$E_2\ (1,2)$} (E);
\draw[->] (E) to [out=-120,in=-60] node[midway, above] {$\emptyset \ (1,3)$} (D);
\draw[->] (E) to [out=-110,in=-70, looseness=2] node[midway, above] {$E_0 \ (1,3)$} (D);
\draw[->] (E) to [out=-100,in=-80, looseness=3] node[midway, above] {$E_1 \ (1,3)$} (D);
\draw[->] (E) to [out=-90,in=-90, looseness=4] node[midway, above] {$E_2 \ (1,3)$} (D);
\end{tikzpicture}}
\end{example}

We represent any data instance $\D$ as a timed word $\sigma_\D$. For $t_i$ occurring in $\D$, let $E(t_i)$ be the maximal set of $\varepsilon$ from $\Pi$ that hold at $t_i$ in $\D$, and let $\sigma_\D = \big((E(t_1), t_1), \dots, (E(t_n), t_n)\big)$. 

%

\begin{example}\label{ex:data}\em
A data instance $\D$ and its representation as $\sigma_\D$ are shown below:\\
\centerline{
\begin{tikzpicture}
\node[point,  label=below:{\scriptsize $P_0'$}, label=above:{\scriptsize $0$}] at (0,0) {};
    \node[point,  label=below:{\scriptsize $Q'$}, label=above:{\scriptsize $1$}] at (1,0) {};
    \node[point,  label=below:{\scriptsize $P_1'$},label=above:{\scriptsize $1.5$}] at (1.5,0) {};
    \node[point,  label=below:{\scriptsize $P_0'$}, label=above:{\scriptsize $4$}] at (4,0) {};
    \node[point,  label=below:{\scriptsize $P_1'$}, label=above:{\scriptsize $4.5$}] at (4.5,0) {};
    \node[point,  label=below:{\scriptsize $P_1'$}, label=above:{\scriptsize $5$}] at (5,0) {};
    \node[point,  label=below:{\scriptsize $Q'$}, label=above:{\scriptsize $6.5$}] at (6.5,0) {};
    \node at (-1, -.1) {$\D:$};
    \node at (0,-0.8) { $\emptyset$};
     \node at (1,-0.8) { $E_1$};
     \node at (1.5,-0.8) { $E_0$};
     \node at (4,-0.8) { $\emptyset$};
     \node at (4.5,-0.8) { $E_2$};
     \node at (5,-0.8) { $E_2$};
     \node at (6.5,-0.8) { $\emptyset$};
     \node at (-1, -0.8) {$\sigma_\D:$};
\end{tikzpicture}
}
\end{example}

	%
\begin{theorem}\label{thm:automaton}
For any linear $\hMTL\!^\diamondminus$ OMQ $(\Pi,A(x))$, a timestamp $t_i$ is a certain answer over a data instance $\D$ iff there exist a subword $\sigma_\D'$ of $\sigma_\D$ with the last timestamp $t_i$ and a run of $\A_\Pi$ over $\sigma_\D'$ that ends with $A$.
\end{theorem}
\begin{example}\em
Let $(\Pi, P_1(x))$ be an OMQ with $\Pi$ from Example~\ref{ex:automaton}. Then, for $\sigma_\D$ from Example~\ref{ex:data}, we have the run $P_0, P_1, P_0, P_1$ on
$$
(E_1, 1), (E_0, \frac{3}{2}), (\emptyset, 4), (E_2, 5),
$$
and so $5$ is a certain answer to the query over $\D$ from Example~\ref{ex:data}.
\end{example}

One could define metric automata as classical timed automata; however, Theorem~\ref{thm:automaton} does not use them in the standard way as it requires runs on subwords. Whether and how such runs can be captured by timed automata remains to be clarified.
We now use the obtained automaton characterisation of certain answers for linear queries to give better complexity bounds for the case of restricted temporal ranges than the \NL{} bound of Theorem~\ref{start-complexity} $(iii)$.

Call an \MTL-program \emph{range-uniform} if all of its temporal operators have the same constraining range.
\begin{theorem}\label{thm:coreminus}
Range-uniform $\cMTL\!^\diamondminus$-OMQs with ranges of the form $\diamondmin_{\langle 0,r\rangle}$ are \textup{FO}$(<,+)$-rewritable.
\end{theorem}
\begin{proof}
We illustrate the proof by a concrete example. Consider the OMQ $(\Pi, S_1)$ with
$$
\Pi = \{ S_0 \leftarrow B, \ S_1 \leftarrow \diamondminus_{(0, d)} S_0,\ S_2 \leftarrow \diamondminus_{(0, d)} S_1,\ S_3 \leftarrow \diamondminus_{(0, d)} S_2,\ S_1 \leftarrow \diamondminus_{(0, d)} S_3\}.
$$
For such a $\Pi$ the automaton $\A_{\Pi}$ is shown in the picture below on the right-hand side. Using it, we construct the following FO-rewriting  $\Q(x)$ of $(\Pi, S_1)$:
$$
\exists x' \, \big[
B(x') \land \forall y \, \big(  (x' < y \le x) \to \exists y'\, \mathsf{dist}_{< d}(y,y') \land{}
 (\varphi_1(x',x)\lor
\varphi_2(x',x) \lor \varphi_3(x',x) ) \big) \big],
$$
where
\begin{itemize}
\item[--] $\varphi_1(x',x) = (x - x') \in 1 + 3 \mathbb{N}$;
\item[--] $\varphi_2(x',x) = (x - x') \in 2 + 3 \mathbb{N} \land \exists x_1 \, ((x'<x_1\le x) \land \varphi_{+1}(x_1,x'))$;
\item[--] $\varphi_3(x',x) = (x - x') \in 3 + 3 \mathbb{N} \land \varphi_{1+1+1}(x',x)\lor \varphi_{1+2}(x',x)$;
\item[--] $\varphi_{1+2}(x',x) =
\exists x_1\, (( x'<x_1\le x) \land \varphi_{+2}(x_1,x'))$;
\item[--] $\varphi_{1+1+1}(x',x) =
\exists x_1,x_2 \, (( x'<x_1<x_2\le x) \land
\varphi_{+1}(x_1,x')\land \varphi_{+1}(x_2,x')\land{} ((x_2 - x_1) > 1))$;
\item[--] $\varphi_{+k}(z,x') = \mathsf{dist}_{< d}(z,z-k-1)  \land ((z -k-1) \ge x')$, for $k = 1,2$.	
\end{itemize}
Intuitively, to derive $S_1$ at $x$, we need a point $x'$ with
$B(x')$ in the data and a sequence of points $y$
between $x'$ and $x$ without gaps of length
$\ge d$. An example of such a data instance is given below.\\[3pt]
\centerline{\scalebox{0.6}{\includegraphics{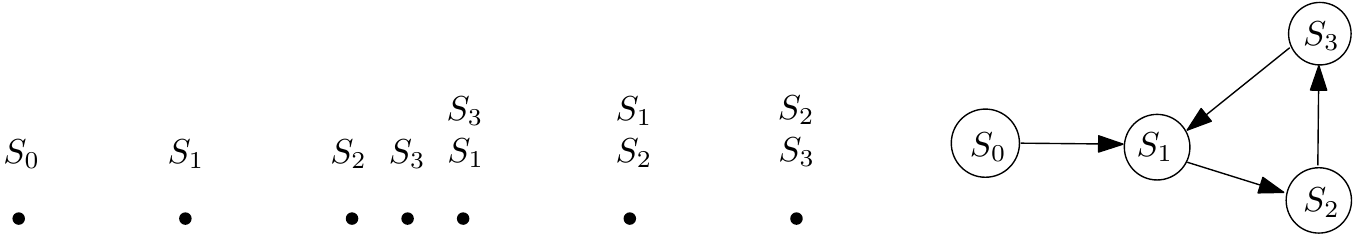}}}

\smallskip
\noindent
Note how we maintain the `stack of states' with the elements at its bottom alternating in a cycle between
$S_1, S_2$, and $S_3$. Note also that the states go in decreasing order when we scan the stack from bottom to top. So we use the formulas $\varphi_k(x',x)$ to express that $S_1$ is inferred at $x$
on level $k$ of the stack. The formula
$\varphi_{+k}(z,x')$ says that the height of the stack increases by $k$ because of a cluster of $k+2$ points within the segment of size $<d$ ending with $z$. The formulas
$\varphi_{1+2}(x',x)$ and $\varphi_{1+1+1}(x',x)$ express two ways of increasing the height of the stack from 1 to 3. It is to be emphasised that properties of $x$ and $ x'$ such as $(x - x') \in 1 + 3 \mathbb{N}$ can be expressed by FO-formulas using the predicate $\text{PLUS}(\textit{num}1, \textit{num}2, \textit{sum})$ or $\text{BIT}(\textit{num},\textit{bit})$, which gives a  binary representation of every object $num$ in the domain of an FO-structure~\cite{Immerman99}, whereas FO with $<$ only is not enough. For example, $(x - x') \in 1 + 3 \mathbb{N}$ is expressed by the formula
$$
\varphi_1(x',x) ~=~  \exists z, z', z'', y\, \big(  (x = y + 1)  \land{}\\ \text{PLUS}(z,z,z') \land{}
\text{PLUS}(z',z,z'') \land
\text{PLUS}(x', z'', y)\big).
$$
We leave further details to the reader.
\end{proof}


\section{OMQs with Punctual Ranges $[r,r]$}\label{sec:punct}

Operators of the form $\diamondminus_{[r,r]}$ resemble the \LTL{} previous time operator $\ominus$. To illustrate an essential difference, consider the program $\Pi = \{\diamondmin_{[1,1]} P \to Q, \   \diamondmin_{[1.5,1.5]} P \land Q \to P\}$
and the data instance $\D$ below. In \LTL, we always derive $\ominus P$ at $n+1$\\[3pt]
\centerline{
\begin{tikzpicture}[nd/.style={draw,thick,circle,inner sep=0pt,minimum size=1.5mm,fill=white},xscale=.3,>=latex]
\draw[thick,->, gray](-1,0) -- (27,0);
\slin{0,0}{$P$}{$0$};
\slin{2,0}{$P$}{$\frac{1}{4}$};
\slin{6,0}{$P$}{$\frac{3}{4}$};
\slin{7,0}{$P$}{$\frac{7}{8}$};
\slin{13,0}{\textcolor{gray}{$P\ Q$}}{$\frac{7}{4}$};
\slin{15,0}{\textcolor{gray}{$Q$}}{$\frac{15}{8}$};
%
\slin{24,0}{}{$3$};
%
\slin{26,0}{$\textcolor{gray}{P} \ Q$}{$\frac{13}{4}$};
\end{tikzpicture}}\\
%
if $P$ holds at $n$. In our example, $P$ at $3/4$ implies $Q$ at $7/4$, which together with $P$ at $1/4$ imply $P$ at $7/4$, and eventually the latter $P$ with $Q$ at $13/4$ implies $P$ at $13/4$; independently, $P$ at $7/8$ implies $Q$ at $15/8$.


\begin{theorem}\label{punc}
\MTL-OMQs with temporal operators of the form $\diamondmin_{[r,r]}$ and $\boxminus_{[r,r]}$ only are \textup{FO(RPR)}-rewritable; answering such OMQs is \NCo-complete for data complexity.
\end{theorem}
\begin{proof}
\NCo-hardness is proved by reduction of \hMTL-OMQs with rules of the form \mbox{$\ominus P \land P' \to Q$}, answering which is NC$^1$-complete~\cite{DBLP:conf/ijcai/ArtaleKKRWZ15}.

To show FO(RPR)-rewritability of a given OMQ $\q=(\Pi,A)$, we assume w.l.o.g.\ that $\Pi$ does not contain ranges $[0,0]$.
Let $R_\Pi$ be the set of  numbers occurring as endpoints of ranges in $\Pi$. We set $\boldsymbol{1} = \text{gcd}(R_\Pi)$, \mbox{$\boldsymbol{n} = \boldsymbol{1}\cdot n$}, for $n \in \mathbb N$, $\boldsymbol{m} = \max (R_\Pi )$. Thus, in our example above, $\boldsymbol{1} = 1/2$, $\boldsymbol{2} = 1$, $\boldsymbol{3} = 3/2$.
We define ${\it cl}(\Pi)$ to be the set of simple and temporal literals with atoms from $\Pi$ and operators $\diamondmin_{\boldsymbol{i}}$ such that $\boldsymbol{i} \in \{\boldsymbol{1}, \dots, \boldsymbol{n}\}$
 and $\diamondmin_{\boldsymbol{n}}$ occurs in $\Pi$. By a \emph{type} $\types$ for $\Pi$ we now mean any maximal subset of ${\it cl}(\Pi)$ consistent with $\Pi$. For types $\types$, $\types'$ and $\boldsymbol{i} \in \{\boldsymbol{1}, \dots, \boldsymbol{m}\}$, we write $\types \rightarrow_{\boldsymbol{i}} \types'$ if
\begin{itemize}
\item[--] $\sigma \in \types$ iff $\diamondmin_{\boldsymbol{i}} \sigma \in \types'$, for any $\diamondmin_{\boldsymbol{i}} \sigma \in {\it cl}(\Pi)$;
\item[--] $\diamondmin_{\boldsymbol{j}} \sigma \in \types$ iff $\diamondmin_{\boldsymbol{j+i}} \sigma \in \types'$, for $\diamondmin_{\boldsymbol{j+i}} \sigma \in {\it cl}(\Pi)$, $\boldsymbol{j} \ge \boldsymbol{1}$.
%
%
%
\end{itemize}
%
We say that $(\types_0, t_0), \dots, (\types_n, t_n)$ is a \emph{run from $t_0$ to $t_n$} on a data instance $\D$ of the form~\eqref{struc} if $t_i \in \tem(\D)$, for $i \le n$, and
%
\begin{itemize}
\item[--] $\{P \in \varSigma_\Pi \mid t_{0} \in P^\D\} \subseteq \types_{0}$;

\item[--] $\neg \diamondmin_{\boldsymbol{j}} \sigma \in \types_0$ for all $\diamondmin_{\boldsymbol{j}} \sigma \in cl(\Pi)$;

\item[--] $\bar t_{i+1} - \bar t_i \in \{\boldsymbol{1}, \dots, \boldsymbol{m}\}$ and if $t_{i+1} > t > t_i$ then $\bar t - \bar t_i \not \in \{\boldsymbol{1}, \dots, \boldsymbol{m}\}$, for any $t \in \tem(\D)$;

\item[--] $\types_i \rightarrow_{(\bar t_{i+1}-\bar t_i)} \types_{i+1}$ and $\{P \! \in\! \varSigma_\Pi \mid t_{i+1} \in P^\D\} \subseteq \types_{i+1}$.
\end{itemize}
Call $t \in \tem(\D)$ \emph{initial} if $\bar t - \bar t' \not \in \{\boldsymbol{1}, \dots, \boldsymbol{m}\}$, for all $t' \in \tem(\D)$. The next lemma follows directly from the given definitions:
\begin{lemma}\label{run}
$(i)$ $(\Pi, \D)$ is consistent iff, for every $t \in \tem(\D)$, there exists a run on $\D$ from some initial $t' \leq t$ to $t$;
$(ii)$ A timestamp $t \in \tem(\D)$ is not a certain answer to $\q$ over $\D$ iff $(\Pi, \D)$ is consistent and
there is a run $(\types_0, t_0), \dots, (\types_n, t_n)$ from initial $t_0$ to $t = t_n$ on $\D$ and $\neg A \in \types_n$.
\end{lemma}
We first show how to express the existence of a run from $x$ to $y$ specified in $(ii)$ by an FO(RPR)-formula $\mathsf{run}_\q(x,y)$ over $\D$. First, as divisibility of binary integers by a given number is recognisable by a finite automaton, we can define an FO(RPR)-formula $\mathsf{div}_{\boldsymbol{1}}(u,v)$ that is true iff $\bar u - \bar v = n \boldsymbol{1}$, for some $n \in \mathbb N$ (see Appendix~\ref{sec:div}). We also have an FO-formula $\mathsf{last}_{\boldsymbol{i}}(u)$ saying that $\boldsymbol{i}$ is minimal among $\{\boldsymbol{1},\dots,\boldsymbol{m}\}$ with $\bar u - \boldsymbol{i} = \bar v$, for some $v \in \tem(\D)$.
Let $Q = \{\types_1,\dots,\types_n\}$ be the set of all types for $\Pi$, and let $Q_0 \subseteq Q$ comprise $\types$ with $\neg \diamondmin_{\boldsymbol{j}} \sigma \in \types$, for all $\diamondmin_{\boldsymbol{j}} \sigma \in {\it cl}(\Pi)$. We define $\mathsf{run}_\q(x,y)$ as the FO(RPR)-formula
\begin{align*}
\left[ \begin{array}{l}
R_{\types_1}(x,z) \equiv \vartheta_{\types_1}\\
\dots\\
R_{\types_n}(x,z) \equiv \vartheta_{\types_n}
\end{array}\right] \bigvee_{\neg A \in\types \in Q} R_{\types}(x,y) ~\land~
\mathsf{div}_{\boldsymbol{1}}(y,x),
\end{align*}
where $R_{\types}(x,z)$, for $\types \in Q$, is a \emph{relation variable}
 %
%
and the formula
$\vartheta_\types(x,z,R_{\types_1}(x,z-1),\dots,R_{\types_n}(x,z-1))$ is a disjunction of the three formulas below if $\types \in Q_0$ and a disjunction of the last two of them if $\types \notin Q_0$:
\begin{align*}
&  (x=z) \land \delta_\types(z) ,\\
&  \neg \mathsf{div}_{\boldsymbol{1}}(z,x)\hspace*{-0.5mm} \land \exists z'\hspace*{-0.5mm} (\mathsf{dist}_{< \boldsymbol{m}}(z,z') \hspace*{-0.5mm}\land\hspace*{-0.5mm} \mathsf{div}_{\boldsymbol{1}}(z',x))\hspace*{-0.5mm} \land \hspace*{-0.5mm}
  R_{\types}(x,z-1),\\
&  \mathsf{div}_{\boldsymbol{1}}(z,x)  \land \hspace*{-2mm} \bigvee_{\substack{\boldsymbol{i} \in \{\boldsymbol{1}, \dots, \boldsymbol{m}\}\\ \types' \rightarrow_{\boldsymbol{i}} \types}}(\delta_\types(z) \land \mathsf{last}_{\boldsymbol{i}}(z) \land R_{\types'}(x,z-1)),
\end{align*}
where $z - 1$ is the immediate predecessor of $z$ in $\tem(\D)$.

To illustrate, in the context of the example above, the formulas $R_\types \equiv \vartheta_\types$ say that $R_{\types}(1/4, 1/4)$ holds for the types
$$
\{\!\neg \diamondmin_{\boldsymbol{1}} P,\! \neg \diamondmin_{\boldsymbol{2}} P, \neg \diamondmin_{\boldsymbol{3}} P, P, Q\},
\{\!\neg \diamondmin_{\boldsymbol{1}} P,\! \neg \diamondmin_{\boldsymbol{2}} P, \neg \diamondmin_{\boldsymbol{3}} P, P, \neg Q\}.
$$
Then $R_{\types}(1/4, 3/4)$ holds for
$$
\{ \diamondmin_{\boldsymbol{1}} P, \neg \diamondmin_{\boldsymbol{2}} P, \neg \diamondmin_{\boldsymbol{3}} P, P, Q\}, \
\{ \diamondmin_{\boldsymbol{1}} P, \neg \diamondmin_{\boldsymbol{2}} P, \neg \diamondmin_{\boldsymbol{3}} P, P, \neg Q\},
$$
$R_{\types}(1/4, 7/8)$ for the same $\types$ as $R_{\types}(1/4, 3/4)$, $R_{\types}(1/4, 7/4)$ for $\types = \{ \neg \diamondmin_{\boldsymbol{1}} P,  \diamondmin_{\boldsymbol{2}} P,  \diamondmin_{\boldsymbol{3}} P, P, Q\}$, and so on.

Thus, we obtain the following FO(RPR)-rewriting of $\q$
$$
\neg \Phi_\Pi \lor \neg \exists y \, \big(\mathsf{run}_\q(y,x) \land \bigwedge_{\boldsymbol{i} \in \{\boldsymbol{1}, \dots, \boldsymbol{m}\}} \neg \mathsf{last}_{\boldsymbol{i}}(y) \big) ,
$$
where $\Phi_\Pi$ checks the consistency condition of Lemma~\ref{run}
$(i)$ and can be constructed similarly to $\mathsf{run}_\q$.
\end{proof}


\section{OMQs with Non-Punctual Ranges}

Unlike the proof of Theorem~\ref{punc}, where the derived facts at $t$ were determined by the data $\D$ at $t$ and the derived facts at the nearest $t'\in \tem(\D)$ with $\bar t' = \bar t - \boldsymbol{i}$, for non-punctual ranges the derived facts at $t$ depend on an unbounded number of timestamps $t' < t$. In the proof of Theorem~\ref{non-punc} below, we show that to construct derivations in this case, we can actually keep track of a fixed number (depending only on the given OMQ)  of moments $t_P' < t$ where each $P$ was derived.


\begin{theorem}\label{non-punc}
$(i)$ \MTL-OMQs whose operators $\diamondmin_\range$ and $\boxminus_\range$ have non-punctual $\range$ are \emph{FO(TC)}-rewritable; answering them is in $\NL$ and $\NCo$-hard;
$(ii)$ \hMTL-OMQs of this kind are \emph{FO(DTC)}-rewritable; answering them is in \emph{L} and $\NCo$-hard.
\end{theorem}
\begin{proof}
In both cases, $\NCo$-hardness can be established as in the proof of Theorem~\ref{punc} by encoding $\ominus$ with $\diamondmin_{(0, 1]}$.

$(i)$ Let $\q = (\Pi, A)$ be the given OMQ.
For $\range = \langle r, q \rangle$ with $q \ne\infty$, let $\range^- = \langle 0, q-r \rangle$ and $\range^+ = \langle 0, q\rangle$; if $q = \infty$, $\range^-$ and $\range^+$ are undefined. Let $\Sigma_\Pi$ be the set of all $\sigma$ with $\diamondmin_\range\sigma$  in $\Pi$, for some $\range$. For $\sigma \in \Sigma_\Pi$, let $\range_\sigma^-$ ($\range_\sigma^+$) be the intersection (union) of the defined $\range^-$ ($\range^+$) with $ \diamondmin_\range\sigma$ in $\Pi$; if there are no such $\diamondmin_\range\sigma$, $\range_\sigma^-$ and $\range_\sigma^+$ are undefined.
To illustrate, consider the \hMTL-program $\Pi$ with the rules
$$
 \diamondmin_{(2,4]} P \to P, \quad \diamondmin_{[1,2)} P \to P, \quad \diamondmin_{[3,\infty]} Q \to Q.
$$
 Then $\range_P^- = (0,1)$, $\range_P^+ = [0,4]$, and $\range_Q^-$, $\range_Q^+$ are undefined.

For a data instance $\D$, a \emph{trace of length $\ell$} for $t \in \tem(\D)$ is a sequence of intervals $[u_0, s_0], \dots, [u_\ell, s_\ell]$ where either $[u_i,s_i] = [*,*]$ (meaning that this interval is undefined) or  $u_i, s_i \in \tem(\D)$,
$u_0 = s_0$, and $u_1 \leq s_1 < u_2 \leq s_2 < \dots < u_\ell \leq s_\ell \leq t$, assuming that $* < u$, for any $u$.
Thus, for the data instance $\D$ below,\\
\centerline{
\begin{tikzpicture}[nd/.style={draw,thick,circle,inner sep=0pt,minimum size=1.5mm,fill=white},xscale=1.1,>=latex]
\draw[thick,->, gray](-0.5,0) -- (7,0);
\slin{-0.2,0}{$P$}{$\frac{1}{2}$};
\slin{0.3,0}{$P$}{$\frac{5}{4}$};
\slin{1.5,0}{$Q$}{$\frac{5}{2}$};
\slin{2.5,0}{}{$\frac{15}{4}$};
\slin{3.2,0}{}{$5$};
\slin{3.6,0}{}{$\frac{25}{4}$};
\slin{6.5,0}{}{$10$};
\end{tikzpicture}}
$([{\textstyle \frac{1}{2}},{\textstyle \frac{1}{2}}],[*,*],[*,*],[{\textstyle \frac{1}{2}},{\textstyle \frac{5}{4}]},[{\textstyle \frac{5}{2}},{\textstyle \frac{5}{2}}])$ is a trace for $t=5/2$. Intuitively, such a trace stores the most recent $\ell$ intervals preceding $t$ where a simple literal holds at some point, with $[u_0, s_0]$ storing the very first point where the literal holds.
A tuple $(\boldsymbol{t}, (\boldsymbol{tr}_\sigma)_{\sigma \in \Sigma_\Pi}, t)$ is an \emph{extended type} for $t \in \tem(\D)$ if
\begin{itemize}
\item[--]  $\boldsymbol{t}$ is a \emph{type} for $\Pi$ (as in the proof of Theorem~\ref{aco});

\item[--] $\boldsymbol{tr}_\sigma$ is a trace for $t$ of length $\ell_\sigma = \lceil |\range_\sigma^+| /|\range_\sigma^-| \rceil$, where $|\range_\sigma^+|$ and $|\range_\sigma^-|$ denote the end-points of these intervals; if one of the intervals is undefined, $\ell_\sigma = 0$;

\item[--] $\diamondmin_\range \sigma \in \boldsymbol{t}$ iff $\mathsf{int}_\range(t, u_i, s_i)$, for some $[u_i, s_i]$ in  $\boldsymbol{tr}_\sigma$,
 \end{itemize}
 where $\mathsf{int}_\range(t, u, s)$ is true iff $\{\bar t - k \mid k \in \range\} \cap [\bar u, \bar s] \neq \emptyset$ and $u,s \ne *$.
In our example, $\ell_P = 4$, $\ell_Q = 0$, and
the following triples $(\boldsymbol{t}_i, (\boldsymbol{tr}^i_\sigma)_{\sigma \in \Sigma_\Pi}, t_i)$ are extended types for $t_i$:
\begin{align*}
& \type_0=\{ P, \neg Q, \neg\diamondmin_{(2,4]}P, \ \neg\diamondmin_{[1,2)}P, \neg\diamondmin_{[3,\infty)}Q \},\ t_0 = {\textstyle \frac{1}{2}}, \
\trace^0_P = ([{\textstyle \frac{1}{2}},{\textstyle \frac{1}{2}}],[*,*],[*,*],[*,*],[{\textstyle \frac{1}{2}},{\textstyle \frac{1}{2}]}),\ \trace^0_Q = ([*,*]);\\
& \type_1= \{P,\neg Q,\neg\diamondmin_{(2,4]}P,\diamondmin_{[1,2)}P, \ \neg\diamondmin_{[3,\infty)}Q \}, t_1 = {\textstyle \frac{5}{4}}, \
  \trace^1_P = ([{\textstyle \frac{1}{2}},{\textstyle \frac{1}{2}}],[*,*],[*,*],[*,*],[{\textstyle \frac{1}{2}},{\textstyle \frac{5}{4}]}),\ \trace^1_Q = ([*,*]);\\
& \type_2 = \{P, Q, \diamondmin_{(2,4]}P, \ \neg \diamondmin_{[1,2)}P, \ \neg\diamondmin_{[3,\infty)}Q \}, \ t_2 = {\textstyle \frac{5}{2}}, \
  \trace^2_P=([{\textstyle \frac{1}{2}},{\textstyle \frac{1}{2}}],[*,*],[*,*],[{\textstyle \frac{1}{2}},{\textstyle \frac{5}{4}]},[{\textstyle \frac{5}{2}},{\textstyle \frac{5}{2}]} ),\
\trace^2_Q=([{\textstyle \frac{5}{2}},{\textstyle \frac{5}{2}]});\\
& \dots\\
%
& \type_5 = \{P, Q, \diamondmin_{(2,4]}P, \diamondmin_{[1,2)}P, \ \neg \diamondmin_{[3,\infty)}Q \}, \ t_5 = {\textstyle \frac{25}{4}}, \
 \trace^5_P=([{\textstyle \frac{1}{2}},{\textstyle \frac{1}{2}}],[{\textstyle \frac{5}{2}},{\textstyle \frac{5}{2}}],[{\textstyle \frac{15}{4}},{\textstyle \frac{15}{4}}],[5,5],[{\textstyle \frac{25}{4}},{\textstyle \frac{25}{4}]} ), \ \trace^5_Q = ([{\textstyle \frac{5}{2}},{\textstyle \frac{5}{2}]}).
\end{align*}
Intuitively, an extended type records the simple and temporal literals that hold at $t$ (the type $\type$) and also some history of the validity of $\sigma$ (the traces) justifying the presence of $\diamondmin_\range \sigma$ in $\type$. As follows from Lemma~\ref{char3} below, to make correct derivations, this history should keep $\ell_\sigma+1$ intervals. Note that this bound does not apply if punctual intervals are present in $\Pi$, which explains the increase of complexity in Theorem~\ref{start-complexity}.


\begin{lemma}\label{char3}
Let $t_0 < \dots < t_m$ be all the timestamps in $\D$. Then $\Pi$ and $\D$ are consistent iff there exists a sequence $(\boldsymbol{t}_i, (\boldsymbol{tr}^i_\sigma)_{\sigma \in \Sigma_\Pi}, t_i)$ of extended types for $t_i$,
$0 \leq i \leq m$, satisfying the following conditions for $\sigma \in \Sigma_\Pi$\textup{:}
\begin{itemize}
\item[--] $\{P \in \varSigma_\Pi \mid t_{i} \in P^\D\} \subseteq \type_{i}$;
\item[--] if $\sigma \notin \boldsymbol{t}_0$, all $[u_j, s_j]$ in $\boldsymbol{tr}^0_\sigma$ are $[*, *]$; if $\sigma \in \boldsymbol{t}_0$, then $[u_0, s_0] = [u_{\ell_\sigma}, s_{\ell_\sigma}] = [t_0, t_0]$ and $[u_j, s_j] = [*, *]$ for $0 < j < \ell_\sigma$\textup{;}
%
\item[--] if $\sigma \not \in \boldsymbol{t}_i$ and $i >0$, then $\boldsymbol{tr}^i_\sigma = \boldsymbol{tr}^{i-1}_\sigma$\textup{;} if $\sigma \in \boldsymbol{t}_i$, $\boldsymbol{tr}^{i-1}_\sigma = ([u_0, s_0], \dots, [u_{\ell_\sigma}, s_{\ell_\sigma}])$ and $[u,s] = [u_0, s_0]$ when $u_0 \neq *$ and $[u,s] = [t_i, t_i]$ otherwise, then
$\boldsymbol{tr}^i_\sigma = ([u, s], [u_1, s_1], \dots, [u_{\ell_\sigma}, t_i])$ if $\bar t_i - \bar s_{\ell_\sigma} \in \range^-_\sigma$, else
          $\boldsymbol{tr}^i_\sigma = ([u, s], [u_2, s_2], \dots, [u_{\ell_\sigma}, s_{\ell_\sigma}], [t_i, t_i])$.
\end{itemize}
\end{lemma}
\begin{proof}
   $(\Rightarrow)$ For a model $\I$ of $\Pi$ and $\D$, we define $(\boldsymbol{t}_i, (\boldsymbol{tr}^i_\sigma)_{\sigma \in \Sigma_\Pi}, t_i)$ with $\boldsymbol{t}_i = \boldsymbol{t}(t_i)$ as follows.  If there is a minimal \mbox{$t_j \leq t_i$} with $\sigma \in \type(t_j)$, we set $[u_0, s_0]=[t_j, t_j]$ in $\boldsymbol{tr}^i_\sigma$; otherwise $[u_0, s_0]= [*,*]$. Consider a  maximal interval $[u,s]$, $s \leq t_i$, such that $\sigma \in \type(u) \cap \type(s)$  and, for any $t \in [u,s)$, there is $t' \in \tem(\D)$ with $\sigma \in \type(t')$  and $\bar t' - \bar t \in \range_\sigma^-$. Suppose there are $k$ such intervals. Let $k' = \min \{k, \ell_\sigma\}$. We define $\boldsymbol{tr}^i_\sigma$ by making its last intervals equal to $[u_{\ell_\sigma-k'+1}, s_{\ell_\sigma-k'+1}], \dots, [u_{\ell_\sigma}, s_{\ell_\sigma}]$, making its $0$-th interval equal to $[u_0, s_0]$, and  making all the remaining intervals equal to $[*,*]$. One can check that $(\boldsymbol{t}_i, (\boldsymbol{tr}^i_\sigma)_{\sigma \in \Sigma_\Pi}, t_i)$ are as required.

   ($\Leftarrow$) Given a sequence  $(\boldsymbol{t}_i, (\boldsymbol{tr}^i_\sigma)_{\sigma \in \Sigma_\Pi}, t_i)$, $1 \leq i \leq m$, we construct an interpretation $\I$ by making $\sigma$ true at $t_i$ in $\I$ iff $\sigma \in \type_i$ for each simple literal $\sigma$ from $\q$. The conditions on these extended types ensure that $\I$ is a model of $\Pi$ and $\D$.
\end{proof}

We use the characterisation of Lemma~\ref{char3} to construct an FO(TC)-sentence $\Phi_\Pi$ that is true in $\D$ iff $\Pi$ and $\D$ are consistent, for any data instance $\D$. $\Phi_\Pi$ contains tuples of variables $\boldsymbol{x} = \boldsymbol{x}_{\sigma_1}, \dots, \boldsymbol{x}_{\sigma_n}$, for $\{ \sigma_1, \dots, \sigma_n\} = \Sigma_\Pi$, where $\boldsymbol{x}_{\sigma} = \boldsymbol{x}_{0\sigma}, \dots, \boldsymbol{x}_{\ell_\sigma \sigma}$ and $\boldsymbol{x}_{i\sigma} = u_{i\sigma}, s_{i\sigma}$ for intervals in traces $\trace_\sigma$;
$\boldsymbol{x}'$ is the same as $\boldsymbol{x}$ but with primed variables:
$$
\Phi_\Pi ~=~ \exists \boldsymbol{x}, \boldsymbol{x}'\,(\bigvee_{\type \text{ type for $\q$}} \mathsf{first}_\type(\boldsymbol{x}) \land{} [\text{TC}_{t, \boldsymbol{x}, t', \boldsymbol{x}'}\xi(t, \boldsymbol{x}, t', \boldsymbol{x}')](\min, \boldsymbol{x}, \max, \boldsymbol{x}')).
$$
Here, $\mathsf{first}_\type(\boldsymbol{x})$ is an FO-formula saying that $\type$ holds in the first timestamp ($\min$) of $\D$ and $\boldsymbol{x}$ represents $\boldsymbol{tr}^0_\sigma$ for all $\sigma$ by encoding $[*,*]$ as the empty interval $[\max,\min]$. The formula $\xi(t, \boldsymbol{x}, t', \boldsymbol{x}')$ under the transitive closure TC says that there is an extended type for $t$ with the trace given by $\boldsymbol{x}$, that $t'$ is the immediate successor of $t$ in $\tem(\D)$, and there is an extended type for $t'$ whose trace is given by $\boldsymbol{x}'$. We define it as
$$
 \xi(t, \boldsymbol{x}, t', \boldsymbol{x}') ~=~ \suc(t', t) \land \bigvee_{\type'\text{ type for $\q$}} \xi_{\type'}(t, \boldsymbol{x}, t', \boldsymbol{x}'),
$$
with $\xi_{\type'}(t, \boldsymbol{x}, t', \boldsymbol{x}')$ saying that if $(\type, (\boldsymbol{tr}_\sigma)_{\sigma \in \varSigma_\Pi}, t)$ is an extended type for $t$ with $(\boldsymbol{tr}_\sigma)_{\sigma \in \varSigma_\Pi}$ given by $\boldsymbol{x}$, then $(\type', (\boldsymbol{tr}'_\sigma)_{\sigma \in \varSigma_\Pi}, t')$ can be the next extended type with $(\boldsymbol{tr}'_\sigma)_{\sigma \in \varSigma_\Pi}$ given by $\boldsymbol{x}'$:
\begin{align*}
  &\xi_{\type'}(t, \boldsymbol{x}, t', \boldsymbol{x}') ~=~ \mathsf{ext}_{\type'}(t', \boldsymbol{x}') \land \bigwedge_{\sigma \not \in \type'} (\boldsymbol{x}_\sigma = \boldsymbol{x}'_\sigma) \land{} \\
  &
  \bigwedge_{\sigma \in \type'} \bigl[
  \big((u_{0\sigma}> s_{0\sigma}) \to (u_{0\sigma}' = t') \land (s_{0\sigma}' = t')\big) \land{}  ((u_{0\sigma}\leq s_{0\sigma}) \to (\boldsymbol{x}_{0\sigma} = \boldsymbol{x}_{0\sigma}')) \land {}
  \\
  &
  \quad \quad \big(\inr_{\range^-_\sigma}(t', s_{\ell_\sigma}) \to \bigwedge_{1 \leq i < \ell_{\sigma}-1}(\boldsymbol{x}_{i\sigma} = \boldsymbol{x}_{i\sigma}') \land{}  (u_{\ell_\sigma \sigma} = u_{\ell_\sigma \sigma}') \land (s_{\ell_\sigma \sigma}' = t')\big) \land{}
  \\
  &
  \quad \quad \big(\neg \inr_{\range^-_\sigma}(t', s_{\ell_\sigma}) \to \bigwedge_{1 < i \leq \ell_{\sigma}-1}(\boldsymbol{x}_{i \sigma} = \boldsymbol{x}_{{i-1} \sigma}') \land{}  (u_{\ell_\sigma \sigma}' = t') \land (s_{\ell_\sigma\sigma}' = t')\big)\bigr].
\end{align*}
The formula $\mathsf{ext}_\type(t, \boldsymbol{x})$ defines an extended type for $t$ in $\D$:
$$
\mathsf{ext}_\type(t, \boldsymbol{x}) ~=~ \delta_{\type}(t) \land{} \bigwedge_{\diamondmin_\range \sigma \in \type} \big(\bigvee_{0 \leq i \leq \ell_\sigma} \mathsf{int}_\range(t, u_{i\sigma}, s_{i\sigma})\big) \land{}
\bigwedge_{\diamondmin_\range \sigma \notin \type} \big(\bigwedge_{0 \leq i \leq \ell_\sigma} \neg \mathsf{int}_\range(t, u_{i\sigma}, s_{i \sigma})\big).
$$
Finally, $\mathsf{first}_\type(\boldsymbol{x})$ is $\bot$ if there is $\diamondmin_\range \sigma \in \type$ and otherwise it is
\begin{align*}
&   \delta_\type(\min) \land  \bigwedge_{\sigma \not \in \type} \bigwedge_{0 \leq i \leq \ell_\sigma} ((u_{i\sigma} = \max) \land (s_{i\sigma} = \min)) \land{}\\
&   \bigwedge_{\sigma \in \type}( \bigwedge_{0 < i < \ell_\sigma} (u_{i\sigma} = \max) \land (s_{i\sigma} = \min)) \land{}
 (u_{0\sigma} = \min) \land (s_{0\sigma} = \min) \land{} (u_{\ell_\sigma \sigma} = \min) \land (s_{\ell_\sigma \sigma} = \min)),
\end{align*}
saying that the intervals in the initial extended type are set correctly. That $\Phi_\Pi$ is as required follows from Lemma~\ref{char3}.

One can now modify $\Phi_\Pi$ to obtain an FO(TC)-rewriting of $\q$.
As before, to obtain a rewriting $\Phi_\q(x)$, we need a formula $\Phi_{\neg A}(x)$ that holds true on $\D$ iff there exists a model of $(\Pi, \D)$ such that $\neg A$ is true at $x$ in this model. We define $\Phi_{\neg A}(x)$ as a conjunction of $\Phi_\Pi$ and
\begin{multline*}
\exists \boldsymbol{x}, \boldsymbol{x}', \boldsymbol{x}'', t_1\, \bigl(((x = 0) \land  \bigvee_{\substack{\type \text{ type for }\Pi\\\neg A \in \type}} \mathsf{first}_\type(\boldsymbol{x})) \lor
(\mathsf{first}_\type(\boldsymbol{x}) \land [\text{TC}_{t, \boldsymbol{x}, t', \boldsymbol{x}'}\theta(t, \boldsymbol{x}, t', \boldsymbol{x}')](\min, \boldsymbol{x}, t_1, \boldsymbol{x}') \land{} \\
\mathsf{suc}(x,t_1) \land \bigvee_{\substack{\type'\text{ type for $\Pi$}\\\neg A \in \type'}} \theta_{\type'}(t_1, \boldsymbol{x}', x, \boldsymbol{x}''))\bigr).
\end{multline*}
The negation of $\Phi_{\neg A}(x)$ is the required rewriting $\Phi_\q(x)$.

$(ii)$ It will be convenient to assume a restricted version of our normal form~\eqref{eq:rule}, where $\diamondmin_\range$ operators do not occur with $\range = [0, r\rangle$. Every $\hMTL$ program $\Pi$ can be converted to this form by replacing $\diamondmin_{[0,r\rangle} A$ by $A \lor \diamondmin_{(0,r\rangle} A$ and then expressing, e.g., $A \lor \diamondmin_{(0,r\rangle} A \to B$ by a pair of rules $A \to B,\, \diamondmin_{(0,r\rangle} A \to B$. (For each $\neg \diamondmin_{[0,r\rangle} \neg A$ we substitute it by $A \land \neg \diamondmin_{(0,r\rangle} \neg A$.) Let $(\boldsymbol{tr}_\sigma)_{\sigma \in \Sigma_\Pi}$ be a trace for $t' \in \tem(\D)$ and $\Lambda$ be the set $\diamondmin_\range \sigma$ from $\Pi$ such that $int_\range(t, u_i, s_i)$, for some $[u_i, s_i]$ in  $\boldsymbol{tr}_\sigma$. Let $\varDelta$ be a set of $P$ from $\Pi$.
We call a type $\type$ for $\Pi$ \emph{minimal for} $t$ with respect to $\varDelta$ \emph{and} $(\boldsymbol{tr}_\sigma)_{\sigma \in \Sigma_\Pi}$ if every $\vartheta$ from $\Pi$ is in $\type$ iff $\vartheta$ is in the closure of $\Lambda$ and $\varDelta$ under rules~\eqref{eq:rule}. We say that a type $\type$ is \emph{minimal initial} with respect to $\varDelta$ if it is minimal for some (every) $t$ with respect to an empty trace with all $[u_j, s_j]$ in $\boldsymbol{tr}_\sigma$ are $[*, *]$ for all $\sigma \in \Sigma_\Pi$.
\begin{lemma}
Let $t_0 < \dots < t_m$ be the timestamps in $\D$. Then $\Pi$ and $\D$ are consistent iff there exists a sequence $(\boldsymbol{t}_i, (\boldsymbol{tr}^i_\sigma)_{\sigma \in \Sigma_\Pi}, t_i)$ of extended types for $t_i$, $0 \leq i \leq m$, satisfying the conditions of Lemma~\ref{char3} and such that: $(i)$ $\type_0$ is minimal initial w.r.t. $\D(t_0)$, $(ii)$ $\type_i$ is the minimal for $t_i$ w.r.t. $\D(t_i)$ and $(\boldsymbol{tr}^{i-1}_\sigma)_{\sigma \in \Sigma_\Pi}$.
\end{lemma}
Note that each type $\type_0$ in the lemma above is uniquely determined by $\D$ and so is the trace $(\boldsymbol{tr}^0_\sigma)_{\sigma \in \Sigma_\Pi}$. Then type $\type_1$ is uniquely determined by $(\boldsymbol{tr}^0_\sigma)_{\sigma \in \Sigma_\Pi}$ and $\D$ and so is $(\boldsymbol{tr}^1_\sigma)_{\sigma \in \Sigma_\Pi}$, etc. Therefore, we can replace in $\Phi_\Pi$ the formula $\xi(t, \boldsymbol{x}, t', \boldsymbol{x}')$ by $\xi'(t, \boldsymbol{x}, t', \boldsymbol{x}')$ such that for given values of $t$ and $\boldsymbol{x}$, there are \emph{unique} values of $t'$ and $\boldsymbol{x}'$ for which $\xi'(t, \boldsymbol{x}, t', \boldsymbol{x}')$ holds.
\end{proof}


\section{Conclusion}

In this paper, we made a first step towards understanding the data complexity of answering queries mediated by ontologies with \MTL{} operators and their rewritability into standard database query languages. By imposing natural restrictions on the ranges $\range$  constraining the operators $\diamondmin_\range$ and $\boxminus_\range$, and by  distinguishing between arbitrary, Horn and core ontologies, we identified classes of \MTL-OMQs that are rewritable to FO$(<)$, FO$(<,+)$, FO(RPR), FO(DTC), FO(TC), and datalog(FO). Unrestricted \MTL-OMQs were shown to be \coNP-hard. The rewritability results look encouraging, though much remains to be done to make our rewritings practical, especially in the presence of more expressive atemporal (description logic or datalog) ontologies.

We can extend our language with constrained operators since $\mathcal{S}_\range$. In this case, \hMTL{} remains \PTime-complete  (but \cMTL{} becomes \PTime-hard) and Theorem~\ref{non-punc} holds, too.
We believe that our \hMTL{} can also be extended with $\boxminus_\range$ in the rule heads (cf.~\cite{DBLP:journals/tcs/Brzoska98}): Theorems~\ref{start-complexity}~$(ii)$ and ~\ref{non-punc}~$(i)$ also hold in this case, but so far we have not managed to prove Theorem~\ref{non-punc}~$(ii)$ for such rules.
Extending \MTL{} with future-time operators is also interesting, in which case Theorems~\ref{start-complexity} and~\ref{aco} remain to hold.
%
%
Finally, we are looking into \MTL-OMQs under the continuous (state-based) semantics, where the techniques developed above do not apply directly.



\medskip
\noindent
\textbf{Acknowledgements} This work was supported by EPSRC UK grant EP/S032282, NCN Poland grant 2016/23/N/HS1/ 02168, and the Foundation for Polish Science (FNP). We would like to thank Stanislav Kikot for his help with the proof of Theorem~5.


\bibliographystyle{plain}
\bibliography{MTL-BIB}

\appendix


\section{$\mathsf{dist}_{=r}(x,y)$ and Related Formulas}\label{sec:dist}

We show that for every $r \in \mathbb{Q}^{\ge 0}_2 $ we can define an FO-formula $\mathsf{dist}_{=r}(x,y)$ that holds in $\mathcal{D}$ iff $x,y \in \varTheta$ and $ \bar x - \bar y =r $.
Let $r \in \mathbb{Q}^{\ge 0}_2$ and $h,k \in \mathbb{N}$ be such that $r= h/2^k$.
Then, we define:
\begin{align*}
\mathsf{dist}_{=r}(x,y) = \;
&\forall j\, \Big(\big(\mathsf{bit}_{\it in}(x, j, 0) \land \mathsf{bit}_{\it in}^{+ h /2^k}(y, j, 0)\big) \lor \big(\mathsf{bit}_{\it in}(x, j, 1) \land \mathsf{bit}_{\it in}^{+ h /2^k}(y, j, 1)\big)\Big) \land {}
\\
&\forall j\, \Big(\big(\mathsf{bit}_{\it fr}(x, j, 0) \land \mathsf{bit}_{\it fr}^{+ h /2^k}(y, j, 0)\big) \lor \big(\mathsf{bit}_{\it fr}(x, j, 1) \land \mathsf{bit}_{\it fr}^{+ h /2^k}(y, j, 1)\big)\Big),
\end{align*}
where  $\mathsf{bit}_{\it in}^{+ h /2^k}(y, j, v)$ states that $v$ is the $j$-th
bit of the integer part of $\bar y + h/2^k$, and $\mathsf{bit}_{\it fr}^{+ h
/2^k}(y, j, v)$ states that $v$ is the $j$-th bit of the fractional part of $\bar
y + h/2^k$. We define $\mathsf{bit}_{\it in}^{+ h /2^k}(y, j, v)$ and
$\mathsf{bit}_{\it fr}^{+ h /2^k}(y, j, v)$ inductively by means of the following
FO($<$)-formulas, where $\ell$ is the last (maximal) element of the domain
$\Delta$ of $\mathcal{D}$, $d\in\mathbb{Q}^{\geq 0}_2$, and $u = \ell - k$ (which
can can be easily defined using $<$):
\begin{align*}
    \mathsf{bit}_{\it fr}^{+ 0 }(y,  j, v) =   \mathsf{bit}_{\it fr}&(y, j, v) ,
    \\
    \mathsf{bit}_{\it fr}^{ +d+(1/2^k) }(y, j, v)  =  \exists u \Bigl( & (u = \ell - k) \land \Bigl( \bigl((j \leq u) \land \mathsf{bit}_{\it fr}^{ + d }(y, j, v) \bigr) \lor{} \\
     &\bigl( (v = 0) \land \mathsf{bit}_{\it fr}^{ + d  }(y, j, 0) \land \exists j' ((u < j' < j) \land \mathsf{bit}_{\it fr}^{ + d }(y, j', 0)) \bigr) \lor{} \\
  &  \bigl( (v = 0) \land \mathsf{bit}_{\it fr}^{  +d}(y, j, 1) \land \forall j' ((u < j' < j) \to \mathsf{bit}_{\it fr}^{  + d }(y, j', 1))\bigr) \lor{}\\
  & \bigl( (v = 1) \land \mathsf{bit}_{\it fr}^{ + d  }(y, j, 1) \land \exists j' ((u < j' < j) \land \mathsf{bit}_{\it fr}^{ + d }(y, j', 0)) \bigr) \lor{} \\
  & \bigl( (v = 1) \land \mathsf{bit}_{\it fr}^{  + d }(y, j, 0) \land \forall j' ((u < j' < j) \to \mathsf{bit}_{\it fr}^{ + d}(y, j', 1)) \bigr) \Bigr) \Bigr),\\
 \mathsf{bit}_{\it in}^{+0}(y, j, v) = \mathsf{bit}_{\it in}&(y, j, v), 
 \end{align*}
\begin{align*}
 \mathsf{bit}_{\it in}^{+ d +(1/2^k)}(y, j, v) \equiv \exists u \Bigl( & (u = \ell - k) \land \Bigl(  \\
     &\bigl( (v = 0) \land \mathsf{bit}_{\it in}^{ + d }(y, j, 0) \land \exists j' ( ((j' < j) \land \mathsf{bit}_{\it in}^{+d}(y, j', 0)) \lor{} \\
      &\hspace{5cm} ((u < j'   \leq  \ell ) \land \mathsf{bit}_{\it fr}^{+d }(y, j', 0))) \bigr) \lor{} \\
  &  \bigl( (v = 0) \land \mathsf{bit}_{\it in}^{ + d }(y, j, 1) \land \forall j' (((j' < j) \to \mathsf{bit}_{\it in}^{+ d}(y, j', 1)) \land{} \\
  & \hspace{5cm} (u < j' <  \leq  \ell  ) \to \mathsf{bit}_{\it fr}^{ + d } (y, j', 1)) \bigr) \lor{}\\
  & \bigl( (v = 1) \land \mathsf{bit}_{\it in}^{ +d }(y, j, 0) \land \exists j' ( ((j' < j) \land \mathsf{bit}_{\it in}^{+ d  }(y, j', 0)) \lor{} \\
      &\hspace{5cm} ((u < j' <  \leq  \ell  ) \land \mathsf{bit}_{\it fr}^{ + d }(y, j', 0))) \bigr) \lor{} \\
  &  \bigl( (v = 1) \land \mathsf{bit}_{\it in}^{ +d }(y, j, 1) \land \forall j' (((j' < j) \to \mathsf{bit}_{\it in}^{ + d }(y, j', 1)) \land{} \\
  & \hspace{5cm}  ((u < j' <  \leq  \ell  ) \to \mathsf{bit}_{\it fr}^{+ d }(y, j', 1))) \bigr) \Bigr) \Bigr).
\end{align*}
The formulas $\mathsf{dist}_{<r}(x,y)$ for $r \in Q_2^{\geq 0} \cup \{\infty\}$
and $\mathsf{dist}_{> r}(x,y)$ are defined by modifications of
$\mathsf{dist}_{=r}(x,y)$. Using these, we can further define FO-formulas
$\inr_\range(x,y)$ and $\mathsf{int}_\range(t, u, s)$. 
%

\section{Divisibility and $\mathsf{div}_d$}\label{sec:div}

We will show how to define an FO(RPR)-formula $\mathsf{div}_d(x,y)$ that is true
in $\D$ iff $u,v \in \varTheta$ and $\bar x - \bar y = d$.

Let $\D$ be an arbitrary  FO-structure.
First, we will  define FO-formulas $b_{fr}(x,y,i)$, $b_{in}(x,y,i)$, $\mathsf{dif}_{fr}(i,x,y)$, and $\mathsf{dif}_{in}(i,x,y)$ such that:
\begin{itemize}
\item[--]  $b_{fr}(x,y,i)$ is true in $\D$ iff when using the column method to
    subtract  $\bar y$ from $\bar x$, the $i$-th bit from the fractional part of
    $\bar x$ is borrowed;
\item[--]  $b_{in}(x,y,i)$ is true in $\D$ iff when using the column method to
    subtract  $\bar y$ from $\bar x$, the $i$-th bit from the integral part of
    $\bar x$ is borrowed;
\item[--] $\mathsf{dif}_{fr}(i,x,y)$ is true in $\D$ iff the $i$-th bit of the
    fractional  part of $\bar x- \bar y$ is $1$;
\item[--] $\mathsf{dif}_{in}(i,x,y)$ is true in $\D$ iff the $i$-th bit of the
    integral part of $\bar x- \bar y$ is $1$.
\end{itemize}

Let the binary representations of fractional parts of $\bar x$ and $\bar y$, be
$x_\ell \dots x_0$ and  $y_\ell \dots y_0$, respectively. Let $b_i \in \{0,1 \}$
indicate whether a bit is borrowed from $x_i$ when subtracting  $y$ from $x$ using
the column method. Clearly, $b_0=0$ and the value of $b_i$ for $i\neq 0$ can be
determined as follows: 
$$
 b_i \equiv (\neg x_{i-1} \land y_{i-1}) \lor (b_{i-1} \land (\neg x_{i-1} \lor y_{i-1})).
$$
Using the equivalence above, we define $b_{fr}(x,y,i) $ as follows:
\begin{align*}
 b_{fr}(x,y,i) =  \; &  \exists j ((j < i) \land (bit_{fr}(j, x, 0) \land bit_{fr}(j, y, 1)) \land {}\\
&  \forall k ((j < k < i) \to  (bit_{fr}(k, x, 0) \lor bit_{fr}(k, y, 1)))).
\end{align*}
In what follows we will denote the binary representation of the fractional part of $x-y$ as $z_\ell \dots z_0$, which can be defined as follows (for $0$ and $1$  treated as truth-values):
$$
z_i \equiv (b_i \land  (x_i  \leftrightarrow y_i) ) \lor (\neg b_i \land (x_i \leftrightarrow \neg y_i)).
$$
Thus, we can define $\mathsf{dif}_{fr}(i,x,y)$ as:
\begin{eqnarray*}
\mathsf{dif}_{fr}(i,x,y) = \exists v, v' (b_{fr}(x,y,i) \land  bit_{fr}(i, x, v) \land bit_{fr}(i, y, v') \land v = v')  \lor
\\
 \exists v, v' (\neg b_{fr}(x,y,i) \land  bit_{fr}(i, x, v) \land bit_{fr}(i, y, v')  \land v \neq v')
.
\end{eqnarray*}
In a similar way we can define the formulas $b_{in}(x,y,i)$ and $\mathsf{dif}_{in}(i,x,y)$ which are about the integral parts of $x$ and $y$.
In particular, we define $b_{in}(x,y,i)$ as:
\begin{eqnarray*}
 b_{in}(x,y,i) =  \bigl(\exists j ((j < i) \land (bit_{in}(j, x, 0) \land bit_{in}(j, y, 1)) \land  \forall k ((j < k < i) \to \\
 (bit_{in}(k, x, 0) \lor bit_{in}(k, y, 1))))\bigr) \lor \\
 \bigl(br_{fr}(x,y,\ell) \land ((bit_{fr}(\ell, x, 0) \lor bit_{fr}(\ell, y, 1))) \land \forall k ((k < i) \to \\
 (bit_{in}(k, x, 0) \lor bit_{in}(k, y, 1)))\bigr),
\end{eqnarray*}
where $\ell$ is (a constant for) the last element of $\varDelta$.
The formula $\mathsf{dif}_{in}(i,x,y)$ is defined analogously to  $\mathsf{dif}_{fr}(i,x,y)$.

\medskip

Next, we will make use of the integer divisibility automaton $\mathcal A_k = (Q,
\{0,1\}, q_0, q_a, \delta)$, which is an NFA taking as an input an inverted binary
representation $z_0 z_1 \dots z_n$ of an integer number $z$ and reaching the
accepting state $q_a$ iff $z$ is divisible by $k$. It is known that for any
integer $k$ we can construct such an automaton. Recall that $z_\ell \dots z_0$ is
the fractional part of $\bar x- \bar y$, i.e., $z_i = 1$ iff
$\mathsf{dif}_{fr}(i,x,y)$ is true in $\D$. Analogously, we denote the integral
part of $\bar x- \bar y$ by $w_{\ell} w_{\ell-1} \dots w_0$, i.e., $w_i=1$  iff
$\mathsf{dif}_{in}(i,x,y)$ is true in $\D$.

Next, we claim that for any $n\in \mathbb{N}$, $k \in \mathbb{Z}$, and a state $q$  of the divisibility automaton $\mathcal A_k$, we can construct an FO-formula $reach_{q, \mathcal A_k}^n(x,y)$ which is true in $\D$ iff $x,y \in \varTheta$ and either:
\begin{itemize}
  \item[--] $\ell \geq n$, $z_i = 0$ for $i < \ell - n$, and $\mathcal A_k$ has a run from $q_0$ to $q$ on $z_{\ell-n} \dots z_{\ell-1} z_{\ell}$; or

\item[--] $\ell < n$ and $\mathcal A_k$ has a run from $q_0$ to $q$ on
      $\underbrace{ 0 \dots 0}_{n-\ell} z_0 \dots z_{\ell}.$
\end{itemize}
To construct $reach_{q, \mathcal A_k}^n$ one needs to consider paths of length bounded by $n$ in $\mathcal A_k$, whose number is finite, and therefore the formula is constructible in FO (we leave details to the reader).

Let $f_d$ be the number of significant bits in the fractional part of the binary representation of $d$ (e.g., $f_d = 3$ for $d = 10001.101$).
Then, we can prove the following result:

\begin{lemma}
Let $d \in \mathbb{Q}_2^{\geq 0}$, $D = d 2^{f_d}$, and let $\mathcal A_D = (Q,
\{0,1\}, q_0, q_a, \delta)$ be  the divisibility automaton for $D$. Then, for any data instance $\D$ and $x, y \in \Theta$, the value of $\bar x- \bar y$ is divisible by $d$ iff there exists $q
\in Q$ such that $reach_{q, \mathcal A_D}^{f_d}(x,y)$ is true in $\D$, and   $\mathcal A_D$ has a run from $q$ to $q_a$  on $w_{0}
w_{1} \dots w_{\ell}$, where $w_{l},  \dots w_1 w_{0} $ is the binary
representation of the integral part of $\bar x- \bar y$.
\end{lemma}
%
Now, let $d$, $D$, and  $\mathcal A_D = (Q, \{0,1\}, q_0, q_a, \delta)$ be as
stated in the lemma. For every $q \in Q$ we will introduce  an FO(RPR)-formula
expressing that $R_q(i,x,y)$ is true in $\D$ iff either:
\begin{itemize}
\item[--]  $i=0$  and there exists $q' \in Q$ such that $reach_{q', \mathcal A_D}^{f_d}(x,y)$ and $q \in \delta(q',w_0)$; or

\item[--] $i>0$ and there exists $q' \in Q$ such that $R_{q'}(i-1,x,y)$ is true in $\D$ and $q \in \delta(q',w_i)$.
\end{itemize}
This formula, denoted by  $\alpha_q$, is as follows:
\begin{align*}
R_q(i,x,y) \equiv & \bigl((i = 0) \land \mathsf{dif}_{in}(0, x,y) \land \displaystyle\bigvee_{q \in \delta(q',1)
} (reach_{q', \mathcal A_D}^{f_d}(x,y))\bigr) \lor {} \\
& \bigl((i = 0) \land \neg \mathsf{dif}_{in}(0, x,y) \land \displaystyle\bigvee_{
q \in \delta(q',0)
} (reach_{q', \mathcal A_D}^{f_d}(x,y))\bigr) \lor {} \\
& \bigl( \mathsf{dif}_{in}(i, x,y) \land \displaystyle\bigvee_{
q \in \delta(q',1)
} R_{q'}(i-1, x, y)\bigr) \lor {} \\
& \bigl( \neg \mathsf{dif}_{in}(i, x,y) \land \displaystyle\bigvee_{
q \in \delta(q',0)
} R_{q'}(i-1, x, y)\bigr).
\end{align*}
Finally, we define $\mathsf{div}_d(x,y)$ by means of the following FO(RPR)-formula, where $q_0, \ldots , q_n$ are all states in $Q$:
\begin{align*}
 \mathsf{div}_d(x,y) =
\left[ \begin{array}{l}
\alpha_{q_0}\\
\dots\\
\alpha_{q_n}
\end{array}\right] R_{q_a}(\ell, x,y)
\end{align*}
Intuitively, the formula uses simultaneous recursion to check whether the accepting state is reached on the input $w_0 w_1 \dots w_{\ell}$.

\end{document}